\documentclass[journal,onecolumn]{IEEEtran}
\ifCLASSINFOpdf
\else
\fi

\usepackage{amsthm}

\usepackage{amssymb}
\usepackage{mathtools}
\usepackage{cases}
\usepackage{autobreak}

\usepackage{booktabs}
\usepackage{multirow}

\usepackage{color}
\usepackage{hyperref}
\usepackage{cite}

\usepackage{amsmath}
\usepackage{extarrows}

\usepackage{comment}

\usepackage{algorithmic}
\usepackage{algorithm}

\usepackage{bbm}

\usepackage{pdfpages}
\usepackage{subfig}

\newcommand{\E}{\mathbb{E}}

\newcommand{\Var}{\mathrm{Var}}
\renewcommand\Re{\operatorname{Re}}

\newcommand{\diag}{\mathrm{diag}}

\newcommand{\LR}{{\bold{R}^{\frac{1}{2}}}}
\newcommand{\RT}{{\bold{T}^{\frac{1}{2}}}}
\newcommand{\MS}{{\bold{S}^{\frac{1}{2}}}}

\newcommand{\FR}{{\bold{R}}}
\newcommand{\FT}{{\bold{T}}}
\newcommand{\FS}{{\bold{S}}}
\newcommand{\BQ}{{\bold{Q}}}
\newcommand{\BH}{{\bold{H}}}
\newcommand{\BI}{{\bold{I}}}
\newcommand{\BX}{{\bold{X}}}
\newcommand{\BE}{{\bold{E}}}
\newcommand{\BZ}{{\bold{Z}}}
\newcommand{\BY}{{\bold{Y}}}

\newcommand{\BG}{{\bold{G}}}

\newcommand{\BU}{{\bold{U}}}

\newcommand{\BA}{{\bold{A}}}
\newcommand{\BB}{{\bold{B}}}
\newcommand{\BC}{{\bold{C}}}
\newcommand{\BD}{{\bold{D}}}
\newcommand{\BT}{{\bold{T}}}
\newcommand{\BS}{{\bold{S}}}

\newcommand{\BF}{{\bold{F}}}

\newcommand{\BO}{{\mathcal{O}}}
\newcommand{\Ba}{{\bold{a}}}
\newcommand{\Bh}{{\bold{h}}}
\newcommand{\By}{{\bold{y}}}

\DeclareMathOperator{\Tr}{Tr}

\newcommand{\RNum}[1]{\uppercase\expandafter{\romannumeral #1\relax}}

\newtheorem{theorem}{Theorem}
\newtheorem{lemma}{Lemma}
\newtheorem{proposition}{Proposition}

\begin{document}
%
\title{Fundamental Limits of Non-Centered Non-Separable
Channels and Their Application in Holographic MIMO Communications}
%
%
%

\author{Xin~Zhang,~{\textit{Graduate Student Member,~IEEE}},~Shenghui~Song,~\IEEEmembership{Senior Member,~IEEE}, and~Khaled B. Letaief,~\IEEEmembership{Fellow,~IEEE}
\thanks{The authors are with the Department of Electronic and Computer Engineering, The Hong Kong University of Science and Technology, Hong Kong
(e-mail: xzhangfe@connect.ust.hk; \{eeshsong, eekhaled\}@ust.hk).}}

\maketitle


\begin{abstract} 
The classical Rician Weichselberger channel and the emerging holographic multiple-input multiple-output (MIMO) channel share a common characteristic of non-separable correlation, which captures the interdependence between transmit and receiver antennas. However, this correlation structure makes it very challenging to characterize the fundamental limits of non-centered (Rician), non-separable MIMO channels. In fact, there is a dearth of existing literature that addresses this specific aspect, underscoring the need for further research in this area. In this paper, we investigate  the mutual information (MI) of non-centered non-separable MIMO channels, where both the line-of-sight and non-line-of-sight components are considered. By utilizing random matrix theory (RMT), we set up a central limit theorem for the MI and give the closed-form expressions for its mean and variance. The derived results are then utilized to approximate the ergodic MI and outage probability of holographic MIMO channels. Numerical simulations validate the accuracy of the theoretical results.
\end{abstract}


%
\IEEEpeerreviewmaketitle

%
%
%
%
\section{Introduction}
Multiple-input multiple-output (MIMO) technology has attracted tremendous interests due to its capability to improve the spectral efficiency and reliability of wireless communications. The fundamental limits of wireless communication systems, i.e., the ergodic mutual information (EMI) and outage probability, indicate the best performance that can be achieved and play a very important role in system design. To characterize the fundamental limits, we need proper channel models to describe the statistical behavior of the propagation environment. In particular, the EMI and outage probability are determined by various channel parameters, e.g., spatial correlation, line-of-sight (LoS) components, noise, interference, Doppler effects, etc~\cite{gao2009statistical}. 


The spatial correlation between antennas reduces the capacity of MIMO systems~\cite{chuah2002capacity} and tremendous efforts have been put in describing the correlated MIMO channels. The most famous one is the correlated Rayleigh model with the Kronecker correlation structure~\cite{kermoal2002stochastic}, which models the correlation at the transmitter and the receiver separately and is mathematically represented by the separable variance profile~\cite{hachem2008new}. However, the Kronecker structure fails to capture the dependence between the transceivers and is not sufficient to describe some realistic indoor MIMO channels~\cite{ozcelik2003deficiencies}. To this end, the Weichselberger model was proposed and validated for both indoor office and suburban outdoor areas~\cite{weichselberger2006stochastic}. Compared with the Kronecker model, the Weichselberger model provides a more general structure in representing a variety of channel scenarios~\cite{wen2011sum,gao2009statistical}. Specifically, it introduces the non-separable structure to allow for arbitrary coupling between the transmit and receive antennas, which takes the separable structure as a special case~\cite{ozcelik2005makes,wen2011sum}. This general model was later utilized to model intelligent reflecting surface-aided channels~\cite{zheng2023}.

Recently, the non-separable model was also utilized in modeling holographic MIMO systems, which were proposed to fully exploit the propagation characteristics of the electromagnetic channel~\cite{wei2022multi,ji2023extra,wan2023can,wan2023mutual,moustakas2022reconfigurable}. In the electromagnetically large regime, the holographic MIMO channel is approximated by the Fourier plane-wave series expansion obtained by the uniform discretization of the Fourier spectral representation for the stationary electromagnetic random field~\cite{pizzo2022spatial}. The corresponding channel in the angular domain is represented by a random matrix with zero mean and non-separable variance profile. Furthermore, as the LoS component becomes more dominating in the millimeter wave and terahertz bands while the small-scale fading is able to characterize the environments with common propagation properties~\cite{marzetta2018spatially,pizzo2020spatially,pizzo2022fourier}, it is necessary to consider both the LoS and non-line-of-sight (NLoS) components in holographic MIMO channels. However, the fundamental limits of non-centered (with both LoS and NLoS components) and non-separable MIMO channel are not yet available in the literature. 


\subsection{Existing Works}
There have been many engaging results regarding the first order and second order analysis for centered/non-centered, separable/non-separable MIMO channels, as shown in Table~\ref{ilu_intro}. In the following, we briefly introduce these related works. 
\begin{table}
\caption{Summary of Related Works.}
\label{ilu_intro}
\centering
\begin{tabular}{|c|c|c|c|c| }
\hline
&  Centered and separable & Non-centered and separable & Centered and non-separable & Non-centered and non-separable \\
\hline 
First-order analysis  &~\cite{tulino2005impact, wei2022multi, pizzo2022fourier}& ~\cite{dumont2010capacity,zhang2021bias}  &  ---  & ~\cite{hachem2007deterministic,wen2011sum,lu2016free} \\
\hline
Second-order analysis &~\cite{hachem2008new,hu2019central,bao2015asymptotic} &  ~\cite{hachem2012clt}  & ~\cite{hachem2008clt}  & This work
\\
\hline
\end{tabular}
\vspace{-0.3cm}
\end{table}

\textbf{First-order analysis:} The first order analysis has covered several scenarios. For the centered and separable case, Tulino~\textit{et al.} performed the EMI analysis for point-to-point MIMO channels in~\cite{tulino2005impact}. For the non-centered and separable case, Dumont~\textit{et al.} derived the closed-form evaluation for the EMI over Rician Kronecker point-to-point MIMO channels and showed that the convergence rate for the deterministic approximation is $\BO(M^{-1})$, where $M$ is the number of transmit antennas~\cite{dumont2010capacity}. In~\cite{zhang2021bias}, Zhang~\textit{et al.} investigated the EMI of non-centered non-Gaussian MIMO fading channels with separable structure. For the non-centered and non-separable case, Hachem~\textit{et al.} derived the deterministic equivalent for the MI over point-to-point MIMO channels in~\cite{hachem2007deterministic}. Considering the non-centered non-separable (Rician Weichselberger) channels, Wen~\textit{et al.} derived the uplink EMI for multi-user MIMO systems in~\cite{wen2011sum} and Lu~\textit{et al.} evaluated the EMI for multiple access channels with distributed sets of correlated antennas in~\cite{lu2016free}. Although the recent holographic MIMO channels have the general non-separable correlation~\cite{pizzo2022fourier}, existing works only focused on the simple separable correlation case. In~\cite{pizzo2022fourier}, Pizzo~\textit{et al.} analyzed the EMI with the separable correlation structure. In~\cite{wei2022multi}, Wei~\textit{et al.} extended the Fourier model in~\cite{pizzo2022fourier} to a multi-user scenario and derived the lower bound for the spectral efficiency of the system with the separable correlation structure. 

\textbf{Second-order analysis:} The MI distribution for centered and separable (including uncorrelated) MIMO channels was investigated by~Hachem~\textit{et al.}~\cite{hachem2008new}, Bao~\textit{et al.}~\cite{bao2015asymptotic}, Hu~\textit{et al.}~\cite{hu2019central}, through setting up a central limit theorem (CLT) for the MI. For the non-centered and separable case, Hachem~\textit{et al.} set up a CLT for point-to-point MIMO channels with separable
variance profile in~\cite{hachem2012clt}. For the centered and non-separable case, Hachem~\textit{et al.} investigated the asymptotic distribution of MI for MIMO channels with a non-separable variance
profile in~\cite{hachem2008clt}. However, the MI distribution and the CLT for the MI over non-centered and non-separable MIMO channels is not available in the literature. In particular, the MI distribution for holographic MIMO channels has not been investigated in the literature. 




In this paper, we will characterize the fundamental limits of non-centered and non-separable MIMO channels by utilizing the asymptotic random matrix theory (RMT). With the increasing size of MIMO antenna arrays, especially for holographic MIMO channels, RMT becomes a very promising method for the related performance evaluation. In fact, RMT has been widely used in the performance evaluation of large-dimension MIMO systems and shown to be accurate even for small-dimension systems~\cite{hachem2008new,moustakas2022reconfigurable,zhang2023secrecy}. Specifically, we will utilize RMT to investigate the distribution of the MI for non-centered and non-separable MIMO channels by setting up a CLT for the MI and the results will be applied to characterize the outage probability of Rician Weichselberger and holographic MIMO channels. 


\subsection{Challenges}
Due to the non-centered and non-separable channel structure, the evaluation of the asymptotic variance of the MI and proof of the Gaussianity are much involved. Specifically, different from the separable case that requires only a system of two equations to characterize the MI, the non-separable correlation requires a system of $M+N$ equations, where $M$ and $N$ denote the number of antennas at the transmitter and receiver, respectively. Such a large number of parameters makes the EMI and variance evaluation challenging. To show the Gaussianity, we employ the martingale method~\cite{bai2008clt} by decomposing the centered MI into the summation of martingales so that the variance evaluation resorts to the sum of squares of martingales. However, due to the existence of both the LoS component and non-separable variance profile, the evaluation for the sum of squares of martingales is very challenging.

\subsection{Contribution}
The contributions of this paper are summarized as follows. 

\begin{itemize}
\item
By utilizing RMT, we set up the CLT for the MI of non-centered and non-separable MIMO channels. The result generalizes the CLT for the MI of centered separable channels in~\cite{hachem2008clt} by adding the LoS component. Meanwhile, the derived results generalizes the CLT for the MI with non-centered separable channels in~\cite{hachem2012clt} by considering the non-separable variance profile. The theoretical results can be utilized to evaluate the outage probability over both the Rician Weichselberger MIMO channels~\cite{weichselberger2006stochastic} and holographic MIMO channels~\cite{pizzo2022fourier}. 

\item
In the proof of the CLT, we show that the evaluation for the square of martingales resorts to the evaluation for the second-order resolvents terms, which can be characterized by a system of $M+N$ equations. Then, we solve for the terms using Cramer's rule and show that the sum can be approximated by a $\log\det$ term. The evaluation for the second-order resolvents can also be utilized for the finite-blocklength analysis over holographic MIMO channels~\cite{hoydis2015second,zhang2022second}.
\end{itemize}


\subsection{Paper Organization}
The rest of this paper is organized as follows. Section~\ref{sec_model} presents the system model and problem formulation. Section~\ref{com_section} introduces the preliminary results including the approximation for the ergodic MI and gives the main results of this paper including the CLT for the MI, which is also used for the approximation of the outage probability. Section~\ref{proof_main_clt} presents the proof of the main results. The theoretical results are validated by numerical simulations in Section~\ref{sec_simu}. Finally, Section~\ref{sec_con} concludes the paper. 

\textit{Notations}:  The bold upper case letters and bold lower case letters denote the matrix and vector, respectively, and $\mathrm{Re} \left\{ \cdot \right\}$ represents the real part of a complex number. The probability measure is represented by $\mathbb{P}(\cdot)$. The space of $N$-dimensional vectors and $M$-by-$N$ matrices are represented by $\mathbb{C}^{N}$ and $\mathbb{C}^{M \times N}$, respectively. The conjugate transpose and the element-wise square root of matrix $\BA$ are given by $\bold{A}^{H}$ and $\BA^{\circ\frac{1}{2}}$, respectively. The $(i,j)$-th entry of $\bold{A}$ and element-wise product of matrices are denoted by $[\bold{A}]_{i,j}$ and $\odot$, respectively. The trace and the spectral norm of $\bold{A}$ are denoted by $\Tr\bold{A} $ and $\|\bold{A} \|$, respectively. The expectation operator and the cumulative distribution function (CDF) of standard Gaussian distribution are denoted by $\E$ and $\Phi(x)$, respectively. The circularly complex Gaussian and real Gaussian distribution are represented by $\mathcal{CN}$ and $\mathcal{N}$, respectively.  The indicator function is denoted by $\mathbbm{1}_{(\cdot)}$. The the almost sure convergence, the convergence in distribution, and the convergence in probability are denoted by
 $\xrightarrow[N \rightarrow \infty]{{a.s.}} $, $ \xrightarrow[N \rightarrow \infty]{\mathcal{D}} $, and $ \xrightarrow[N \rightarrow \infty]{\mathcal{P}} $ respectively. The Big-O, the Little-o are denoted by $O(\cdot)$ and $o(\cdot)$, respectively. Specifically, $f(n)\in O(g(n))$ if and only if there exists a positive constant $c$ and a nonnegative integer $n_{0}$ such that $f(n) \le cg(n)$ for all $n \ge n_{0}$. $f(n)\in o(g(n))$ if and only if there exists a nonnegative integer $n_{0}$ such that $f(n) \le cg(n)$ for all $n \ge n_{0}$ for all positive $c$~\cite{cormen2009introduction}. Here $[N]$ denotes the set $\{1,2,...,N \}$.



\section{System Model and Problem Formulation}
\label{sec_model}

\subsection{MIMO Communications}
Consider a point-to-point MIMO system with $N_R$ receive antennas and $N_T$ transmit antennas and assume that perfect channel information (CSI) is available at the receiver. The receive signal $\bold{y}\in \mathbb{C}^{N_{R}}$ is given by
\begin{equation}
\label{sig_mod}
\bold{y}=\BH_{s}\bold{x}+\bold{n},
\end{equation}
where $\bold{x}\in \mathbb{C}^{N_{S}} \sim \mathcal{CN}(0,\BI_{N_S})$ and $\bold{n}\in \mathbb{C}^{N_{R}} \sim \mathcal{CN}(0,\bold{I}_{N_{R}})$ represent the transmit signal and the additive white Gaussian noise (AWGN), respectively. Here $\BH_s \in \mathbb{C}^{N_R\times N_T}$ denotes the channel matrix and the subscript $s$ is utilized to represent different channel models.

\subsection{Non-Separable Correlation}
\label{chamodel}
In this part, we introduce two non-separable correlated models, i.e., Rician Weichselberger and holographic MIMO channels, to illustrate the non-separable correlation. Despite the common non-separable correlation structure, these two channel models were proposed with different motivations. The details are given below.

\subsubsection{Rician Weichselberger Channels}
The Rician Weichselberger was proposed to compensate for the Rician Kronecker model by considering the joint correlation between the transmit antennas and receive antennas. Hence, we start by introducing the Rician Kronecker model. The Rician channel with Kronecker correlation model can be represented by~\cite{dumont2010capacity}
\begin{equation}
\label{kron_rician}
\BH_{K}= \BA_{K}+ \LR \BY_{K}\RT,
\end{equation}
where $\FR \in \mathbb{C}^{N_R \times N_R}$, $\FT \in \mathbb{C}^{N_T \times N_T}$, and $\BA_{K} \in \mathbb{C}^{N_R\times N_T}$ denote the correlation matrices at the receiver and transmitter, and the LoS component, respectively. $\BY_{K} \in \mathbb{C}^{N_R \times N_T}$ is a random matrix consisting of independent and identically distributed (i.i.d.) entries with zero mean and variance $N_{T}^{-1}$. The channel model in~(\ref{kron_rician}) can be equivalently represented by~\cite{weichselberger2006stochastic}
\begin{equation}
\label{kron_rician2}
\BH_{K}= \BA_{K}+ \BU_{K} ( \bold{D}^{\frac{1}{2}}\BY_{K}\widetilde{\bold{D}}^{\frac{1}{2}} ) \bold{V}^{H}_{K}= \BA_{K}+ \BU_{K} (  (\bold{d}_{K}\widetilde{\bold{d}}^{T}_{K} )^{\circ \frac{1}{2}} \odot \BX_{K}  ) \bold{V}^{H}_{K},
\end{equation}
where the correlation matrices have the eigen-decomposition $\FR=\BU_{K} \BD_{K}\BU^{H}_{K} $ and $\FT=\bold{V}_{K} \widetilde{\BD}_{K} \bold{V}^{H}_{K} $ with $\bold{d}_K=\diag(\BD_{K})$ and $\widetilde{\bold{d}}_K=\diag(\widetilde{\BD}_{K})$. Here $\BX_{K}=\BU^{H}_{K}\BY_{K} \bold{V}_{K}$ has the same distribution as $\BY_{K}$ due to the unitary invariant attribute of Gaussian matrices. It can be observed from~(\ref{kron_rician2}) that the correlation at the transmitter does not have impact on the receiver side. 

The Weichselberger model was proposed to alleviate the restriction in~(\ref{kron_rician2}) and describe the joint spatial structure of the channel. The Rician Weichselberger Model can be represented by
\begin{equation}
\label{weich_cha}
\BH_{W}=  \BA_{W}+\bold{U}_W\left( \bold{\Sigma}^{\circ \frac{1}{2}}_{W} \odot  \BX_{W} \right) \bold{V}^{H}_W=\bold{U}_W\left( \overline{\BA}_{W}+\bold{\Sigma}^{\circ \frac{1}{2}}_{W} \odot  \BX_{W} \right) \bold{V}^{H}_W=\bold{U}_W\overline{\BH}_{W} \bold{V}^{H}_W,
\end{equation}
where $\overline{\BH}_{W}= \overline{\BA}_{W}+\bold{\Sigma}^{\circ \frac{1}{2}}_{W} \odot  \BX_{W} $, $ \overline{\BA}_{W}=\BU^{H}_W  {\BA}_{W}\bold{V}_W$, $\bold{U}_W \in \mathbb{C}^{N_R\times N_R }$ and $\bold{V}_W \in \mathbb{C}^{N_T\times N_T }$ are deterministic unitary matrices and $\BA_{W} \in \mathbb{R}^{N_R}$ is the LoS component. Matrix $\BX_{W}\in\mathbb{C}^{N_R\times N_T}$ is a random matrix consisting of i.i.d. entries with zero mean and variance $N_{T}^{-1}$. The variance profile matrix $\bold{\Sigma}^{\circ \frac{1}{2}}_{W} \in \mathbb{R}^{N_R\times N_T}$, consisting of non-negative entries, is called the ``coupling matrix'' since $[\bold{\Sigma}^{\circ \frac{1}{2}}_{W}]_{i,j}$ characterizes the energy coupled between the $i$-th eigenvector of $\BU_W$ and the $j$-th eigenvector of $\bold{V}_W$~\cite{kermoal2002stochastic}. 

\textit{Non-separable correlation in Weichselberger model:} Mathematically, the separability is determined by whether the ``coupling matrix'' $\bold{\Sigma}_W$ can be written as $\bold{\Sigma}_{W}=\bold{d}\widetilde{\bold{d}}^{T}$, where $\bold{d} \in \mathbb{R}^{N_R}$ and $\widetilde{\bold{d}} \in \mathbb{R}^{N_T}$. In fact, $\bold{\Sigma}_W$ is generally non-separable and the Kronecker correlation model in~(\ref{kron_rician2}) is a special case of~(\ref{weich_cha}) when $\bold{\Sigma}_{W}= \bold{d}\widetilde{\bold{d}}^{T}$.

\subsubsection{Holographic Channels} Inspired by the very high energy and spectral efficiency of massive MIMO systems~\cite{ngo2013energy,bjornson2017massive}, the dense and electromagnetically large (compared with the wavelength) antenna array, referred to as the ``Holographic array'', was proposed as a promising technology to further boost the performance limits of wireless communications~\cite{marzetta2018spatially}. ``Holographic'', which literately means ``describe everything'' in the ancient Greek, refers to the regime where the MIMO system is designed to fully exploit the propagation characteristics of the electromagnetic channel~\cite{dardari2021holographic}. 


Fig.~\ref{sys_fig} shows the holographic MIMO channel consisting of two parallel planar arrays that are perpendicular to the $z$-axis, where $\bold{s}_{j}=[s_{x_j},s_{y_j},s_{z_j}]^{T}$ and $\bold{r}_i=[r_{x_i},r_{y_i},r_{z_i}]^{T}$ represent the coordinates of the $j$-th transmit antenna and $i$-th receive antenna, respectively. The two planar arrays span the rectangular areas in the $xy$-plane with the size $L_{S,x}\times L_{S,y}$ at the transmit array and $L_{R,x}\times L_{R,y}$ at the receive array. There are $N_S$ transmit and $N_R$ receive antennas and they are deployed uniformly with spacing $\Delta_{S}$ and $\Delta_{R}$ at the transmit and receive array, respectively, which refer to the distance between centers of adjacent antennas. The wave with wavelength $\lambda$ propagates in a homogeneous and infinite medium without polarization. In the following, we will use $\BH_s$ and $\BH_a$ to represent the channel matrix in the spatial and angular domain, respectively, and use $\BH_h$ to denote the discretized version of $\BH_a$.

\begin{figure}[t!]
\centering\includegraphics[width=0.45\textwidth]{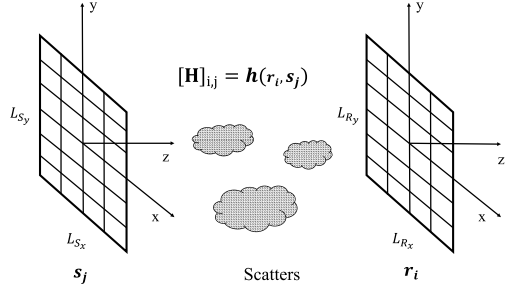}
\caption{Holographic MIMO System with Planar Arrays.}
\label{sys_fig}
\end{figure}






\begin{figure}[t!]
\centering\includegraphics[width=0.45\textwidth]{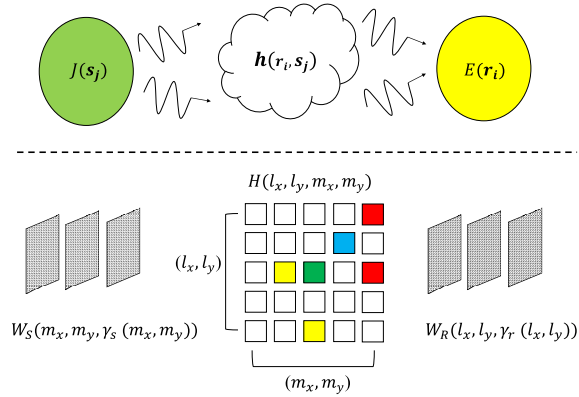}
\caption{Plane-wave representations.}
\label{plane_wave_repre}
\vspace{-0.3cm}
\end{figure}
\begin{figure}[t!]
\centering\includegraphics[width=0.45\textwidth]{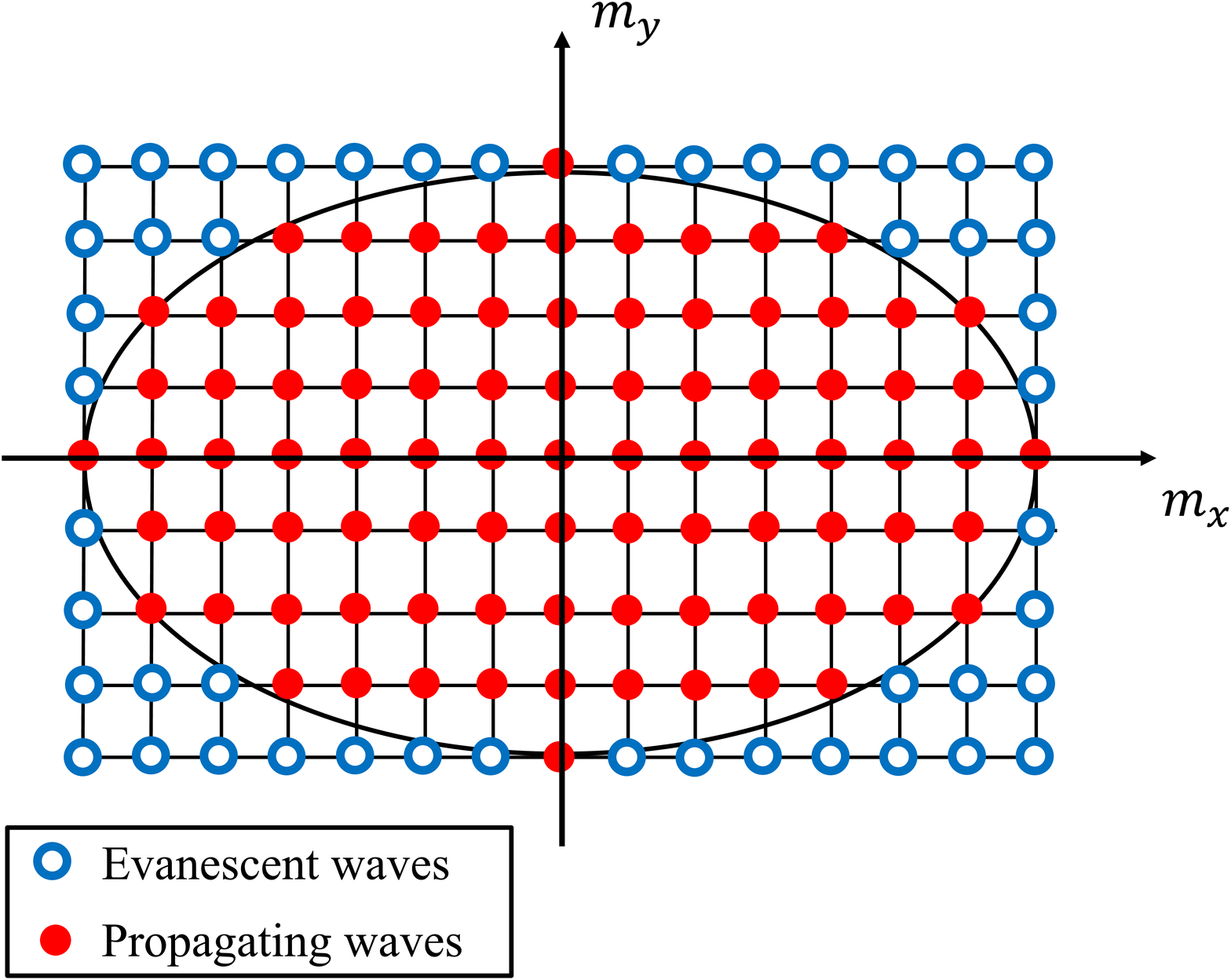}
\caption{Propagating and evanescent waves.}
\label{eva_plane}
\vspace{-0.3cm}
\end{figure}

We focus on the \textit{Fourier plane-wave representation} of holographic MIMO channels~\cite{pizzo2022fourier,pizzo2020holographic}. As shown by Fig.~\ref{plane_wave_repre}, for each pair of receive and transmit antennas $(\bold{r}_i,\bold{s}_j)$, by using Fourier expansion, both the transmit field $J(\bold{s}_j)$ and receive field $E(\bold{r}_i)$ passing through the channel $h(\bold{r}_{i},\bold{s}_j)$ can be represented by plane-waves and parameterized by the horizontal wavenumbers $(\frac{2\pi m_{x}}{L_{S,x}},\frac{2\pi m_{y}}{L_{S,y}})$ and $(\frac{2\pi l_{x}}{L_{R,x}},\frac{2\pi l_{y}}{L_{R,y}})$, respectively. The receive plane-wave coefficients $W_{R}(l_x,l_y,\gamma_R(l_x,l_y))$ can be obtained by taking a four-variable scattering kernel $H(l_x,l_y,m_x,m_y)$ to the transmit plane-wave coefficients $W_{S}(m_x,m_y,\gamma_S(m_x,m_y))$. As shown in Fig.~\ref{eva_plane}, the plane-waves for the transmit field include the propagating waves and evanescent waves, which are parameterized by pairs $(m_x,m_y)$ inside and outside the elliptical area, respectively. Here we assume that the evanescent waves are negligible since they decay exponentially with $z$\footnote{Some researchers think that the evanescent waves can alway be neglected~\cite{franceschetti2017wave} but some think it is not yet settled~\cite{ji2023extra}.}. For LoS and non-LoS channels, the kernel $H(l_x,l_y,m_x,m_y)$ is modeled as deterministic and random matrices, respectively~\cite{pizzo2022fourier,pizzo2020holographic}, and the associated field is stationary horizontally. The mathematical formulation of the channel model is omitted here and given in Appendix~\ref{fourier_details} for interested readers. 



According to~\cite[Theorem 2]{pizzo2022fourier}, in the~\textit{electromagnetically large regime}, i.e.,
\begin{equation}
\label{large_exp}
\min\{\frac{L_{S,x}}{\lambda},\frac{L_{S,y}}{\lambda},\frac{L_{R,x}}{\lambda},\frac{L_{R,x}}{\lambda} \} \gg 1,
\end{equation}
$\BH_s$ can be approximated by the discretized Fourier spectral representation through uniformly sampling $(\kappa_x,\kappa_y)$ at the direction $(\frac{2\pi m_x}{L_{S,x}},  \frac{2\pi m_y}{L_{S,y}})$ and $(k_x,k_y)$ at the direction $(\frac{2\pi l_x}{L_{R,x}},  \frac{2\pi l_y}{L_{R,y}})$, respectively. It is given by
\begin{equation}
\label{hs_fourier}
\begin{aligned}
&[\BH_{s}]_{ij}=h_s(\bold{r}_i,\bold{s}_j)= \sum_{(l_x,l_y)\in \mathcal{E}_{R} }\sum_{(m_x,m_y)\in \mathcal{E}_{S} } 
{H}(l_x,l_y,m_x,m_y)
a_{r}(l_x,l_y,\bold{r}_i)a_{s}^{*}(m_x,m_y,\bold{s}_j),
\end{aligned}
\end{equation}
where the discretized Fourier coefficient ${H}(l_x,l_y,m_x,m_y)$, named as coupling coefficients~\cite{miller2000communicating}, is the angular response and only related to the $xy$-coordinates since $k_z, \kappa_z$ can be determined by $\gamma(k_x,k_y)$ and $\gamma(\kappa_x,\kappa_y)$, respectively. In this case, the region in Fig.~\ref{eva_plane} can be represented by the lattice ellipse $\mathcal{E}_{S}=\{ (m_x,m_y)\in \mathbb{Z}^{2}  | (\frac{m_x}{L_{S,x}})^2+(\frac{m_y}{L_{S,y}})^2  \le 1   \}$ and $\mathcal{E}_{R}=\{ (l_x,l_y)\in \mathbb{Z}^{2}  | (\frac{l_x}{L_{R,x}})^2+(\frac{l_y}{L_{R,y}})^2  \le 1   \}$. The cardinality of $\mathcal{E}_{R}$ and $\mathcal{E}_{S}$ can be evaluated by 
\begin{equation}
\begin{aligned}
n_R & =\lceil  \frac{\pi L_{R,x}L_{R,y} }{\lambda^2}  \rceil +o(\frac{ L_{R,x}L_{R,y} }{\lambda^2}), 
\\
n_S & =\lceil  \frac{\pi L_{S,x}L_{S,y} }{\lambda^2}  \rceil +o(\frac{ L_{S,x}L_{S,y} }{\lambda^2}),
\end{aligned}
\end{equation}
 such that $(l_x,l_y)$ and $(m_x,m_y)$ can be indexed by $[n_R]$ and $[n_S]$, respectively. Different from~\cite{pizzo2022fourier,pizzo2020holographic}, both the LoS and NLoS components are considered in this paper such that the channel matrix in the spatial domain can be represented by
\begin{equation}
\label{H_HA}
\BH_{s}={\BA}_{s}+\BY_{s}=\sqrt{G_R G_S  N_R N_S } \bold{\Phi}_{R}{ \BH}_h\bold{\Phi}_{S}^{H},
\end{equation}
where
\begin{equation}
\label{H_reform}
\BH_h=\BA_h+\BY_h=\BA_h+\bold{\Sigma}^{\circ \frac{1}{2}}_h \odot \BX_h.
\end{equation}
Here, ${\BA}_s\in \mathbb{C}^{N_R\times N_{S}}$ and $\BY_s \in \mathbb{C}^{N_R\times N_{S}}$ denote the LoS and NLoS components, respectively, and ${\BA_h}\in \mathbb{C}^{n_R\times n_{S}}$ and $\BY_h \in \mathbb{C}^{n_R\times n_S}$ are the corresponding representations in the wave number domain. The coefficient of the LoS component can be obtained by~\cite[Eq. (12)]{pizzo2020holographic} and the coefficient of the NLoS component is given by $[\BY_h]_{i,j}=Y(l_x,l_y,m_x,m_y)\sim\mathcal{CN}(0, \frac{\sigma^2(l_x,l_y,m_x,m_y)}{n_{S}})$ such that the coefficient matrix of the small-scale fading can be written as ${\BY_h}= \bold{\Sigma}^{\circ \frac{1}{2}}_h\odot \BX_h
$ with $[\bold{\Sigma}_h]_{i,j}=\sigma^2(l_x,l_y,m_x,m_y)$ denoting the variance profile of the coupling coefficients. Matrix $\BX_h \in \mathbb{C}^{n_R\times n_S}$ is an i.i.d. random matrix with circularly-symmetric complex-Gaussian random entries whose variance is $\frac{1}{n_S}$. $\bold{\Phi}_{R}$ and $\bold{\Phi}_{S}$ are semi-unitary matrices consisting of the Fourier orthogonal bases. The coefficient $\sqrt{N_R N_S}$ rises from the normalization of the basis. $G_S$ and $G_R$ represent the patch antenna gain at the transmitter and receiver, respectively, with
\begin{equation}
\label{antenna_gain}
G=\frac{4\pi \tau S}{\lambda^2},
\end{equation}
where $S$ is the antenna area and $\tau<1$ denotes the antenna efficiency~\cite{balanis2016antenna}. 

\textit{Non-separable correlation in holographic MIMO channels:} In the holographic MIMO systems, the small-scale fading can be also divided into non-separable and separable cases according to the structure of $\bold{\Sigma}$. The variance of the small-scale fading is given by
\begin{equation}
\begin{aligned}
\label{var_A}
&\sigma^2(l_x,l_y,m_x,m_y)=
 \iiiint\limits_{\Omega_{S}(m_x,m_y)\times \Omega_{R}(l_x,l_y)} S^2(\theta_{R},\phi_R,\theta_S,\phi_S) \mathrm{d} \Omega_{R} \mathrm{d} \Omega_{S},
\end{aligned}
\end{equation}
where $(\theta_{R},\phi_R,\theta_S,\phi_S)$ represents the spherical coordinates of $(l_x,l_y,m_x,m_y)$, and $\Omega_{S}(m_x,m_y)$ and $\Omega_{R}(l_x,l_y)$ denote the corresponding areas of $\mathcal{E}_{m_x,m_y}$ and $\mathcal{E}_{l_x,l_y}$ in spherical coordinates, respectively. $S^2(\theta_{R},\phi_R,\theta_S,\phi_S)$ represents the bandlimited 4D power spectral density of $h_{s}(\bold{r},\bold{s})$. In general, $\sigma^2(l_x,l_y,m_x,m_y)$ can not be decomposed as a product of the functions of receive wavenumbers $(l_x,l_y)$ and transmit wavenumbers $(m_x,m_y)$, i.e., $\sigma^2(l_x,l_y,m_x,m_y) \neq \sigma^2_{S}(m_x,m_y)\sigma^2_{R}(l_x,l_y)$. This model is referred to as the~\textit{non-separable profile}. The~\textit{separable profile} refers to the case with
\begin{equation}
\label{sep_def}
\sigma^2(l_x,l_y,m_x,m_y)=\sigma^2_{S}(m_x,m_y)\sigma^2_{R}(l_x,l_y),
\end{equation}
where $\sigma^2_{S}(m_x,m_y)$ and $\sigma^2_{R}(l_x,l_y)$ represent the channel power transfer at the transmitter and receiver, respectively. The separable model can be obtained from the non-separable case by taking $S^2(\theta_{R},\phi_R,\theta_S,\phi_S) =1$. Under such circumstances, $\bold{\Sigma}_h$ can be represented by $\bold{\Sigma}_h=\bold{d}\widetilde{\bold{d}}^{T}$,
where $\bold{d}\in \mathbb{R}^{n_R}=[\sigma^2_{R}(l_{x,1},l_{y,1}),...,\sigma^2_{R}(l_{x,n_R},l_{y,n_R})]^{T}$ and $\widetilde{\bold{d}}\in\mathbb{R}^{n_S}=[\sigma^2_{S}(m_{x,1},m_{y,1}),...,\sigma^2_{S}(m_{x,n_S},m_{y,n_S})]^{T}$. This indicates that $\bold{\Sigma}^{\circ \frac{1}{2}}_h \odot \BX_h= \bold{D}^{\frac{1}{2}}\BX_h\widetilde{\bold{D}}^{\frac{1}{2}}$, where $\bold{D}=\diag(\bold{d})$ and $\widetilde{\bold{D}}=\diag(\widetilde{\bold{d}})$. It can be observed that the separable structure is a special case of the non-separable structure and the performance of non-LoS MIMO channels with separable profile has been evaluated in~\cite{pizzo2022fourier}. However, the fundamental limits of holographic MIMO systems with non-separable correlation structure is not yet available in the literature and will be the focus of this paper.

\subsection{Problem Formulation} 
In this section, we show that the MI for both the Rician Weichselberger model and holographic MIMO channels resorts to the MI with general non-centered non-separable channels. The MI of the MIMO system is given by
\begin{equation}
\label{mi_first}
C(\sigma^2)=\log\det(\bold{I}_{N_R}+\frac{1}{\sigma^2}\BH_{s}\BH_{s}^{H}),
\end{equation}
where $\frac{1}{\sigma^2}$ denotes the signal-to-noise ratio (SNR). For the Rician Weichselberger model in~(\ref{weich_cha}), the MI in~(\ref{mi_first}) can be rewritten as
\begin{equation}
\label{rician_mi}
C(\sigma^2)=\log\det(\bold{I}_{N_R}+\frac{1}{\sigma^2}\overline{\BH}_{W}\overline{\BH}^{H}_{W}),
\end{equation}
where we have utilized the identity $\det(\bold{I}+\BA\BB)=\det(\bold{I}+\BB\BA)$.
For holographic MIMO channels in~(\ref{H_HA}),~(\ref{sig_mod}) can be rewritten as  
\begin{equation}
\bold{y}_{a}=\sqrt{G_R G_S N_R N_S}\BH_h \bold{x}_a+\bold{n}_a,
\end{equation}
where $\bold{y}_a \in \mathbb{C}^{n_R}=\bold{\Phi}_{R}^{H}\bold{y}$ and $\bold{x}_a \in \mathbb{C}^{n_S}=\bold{\Phi}_{S}^{H}\bold{x}$ represent the receive and transmit signal in the angular domain. By the unitary-invariant attribute of the Gaussian random vector, we have $\bold{n}_a \in \mathbb{C}^{n_R}\sim\mathcal{CN}(0,\sigma^2\bold{I}_{n_R})$. By the identity $\det(\bold{I}+\BA\BB)=\det(\bold{I}+\BB\BA)$, the MI in~(\ref{mi_first}) can be rewritten as
\begin{equation}
\label{mi_second}
C(\frac{\sigma^2}{G_{R}G_{S} N_R N_S} )=\log\det(\bold{I}_{n_R}+  \frac{G_{R}G_{S} N_R N_S}{\sigma^2} {\BH}_h { \BH}^{H}_h),
\end{equation}
where $\BH_h \in \mathbb{C}^{n_R \times n_S}$ is given in~(\ref{H_reform}). 

According to~(\ref{rician_mi}) and~(\ref{mi_second}), the MI for both Rician Weichselberger and holographic MIMO models can be generally formulated as
\begin{equation}
\label{mi_third}
C_M(\zeta)=\log\det(\bold{I}_{M}+  \zeta^{-1} {\BH} { \BH}^{H}),
\end{equation}
where $\BH \in\mathbb{C}^{N\times M}=\BA+\bold{\Sigma}^{\circ \frac{1}{2}}\odot \BX$. We have $N=N_R$, $M=N_T$, $\zeta=\sigma^2$ and $N=n_{R}$, $M=n_S$, $\zeta=\frac{\sigma^2}{G_{R}G_{S} N_R N_S}$ for Rician Weichselberger and holographic MIMO models, respectively, such that the MI for both cases resort to the investigation of $C_{M}(\zeta)$ in~(\ref{mi_third}). Here ${\BH}$ is the sum of a deterministic matrix and a random matrix whose entries have heterogeneous variances, and is referred to as the non-centered non-separable channel. Due to the complex structure, the characterization of the MI is a difficult problem. In the following, we will investigate the distribution of $C_{M}(\zeta)$ by RMT in the asymptotic regime where the dimensions of $\BH$ go to infinity with the same pace.

\section{Main Results}
\label{com_section}
In this section, we will first present the assumptions and existing first-order result that will be utilized to derive the asymptotic distribution of the MI. Then, we will give the main results of this paper, which are based on the following assumptions:

\textbf{A.1}: $0<\lim \inf\limits_{M \ge 1}  \frac{M}{N} \le \frac{M}{N}  \le \lim \sup\limits_{M \ge 1}  \frac{M}{N} <\infty$. 

\textbf{A.2}: $0<\min\limits_{i,j} \sigma^2_{i,j}\le \sigma^2_{i,j}\le\max\limits_{i,j} \sigma^2_{i,j}< \sigma^2_{max}<\infty $, $\inf\limits_{M\ge 1} \frac{\sum_{i=1}^{N}\sigma^2_{i,j}}{M}>0$, $\inf\limits_{M \ge 1} \frac{\sum_{j=1}^{M}\sigma^2_{i,j}}{M}>0$.

\textbf{A.3}: $\|\BA \|<a_{max}<\infty$. 

Here $a_{max}$ and $\sigma^2_{max}$ represent the upper bound of $\|\BA \|$ and $\bold{\Sigma}_{i,j}$, which do not go to infinity with $M$ and we use $\sigma^2_{min}$ to represent the lower bound of $\bold{\Sigma}_{i,j}$. \textbf{A.1} indicates the asymptotic regime that the dimensions $M$ and $N$ go to infinity with the same pace. For holographic MIMO systems, by the definition in~(\ref{large_exp}), we know $M$ and $N$ are large and~\textbf{A.1} holds true for the case when the physical size of the planar array of one side (transmit or receive) is not overwhelmingly larger than that of the other side.~\textbf{A.2} is used to guarantee that the asymptotic variance is well defined, which is justified by the boundness of the variance for coefficients~\cite{pizzo2022fourier}.~\textbf{A.3} implies that $\|\bold{a}_i \|_{2}$ is uniformly bounded.~\textbf{A.3} also indicates that the rank of the LoS component $ \BA $ increases with $M$ and $N$ at the same pace~\cite{zhang2021bias}. We denote the limit $M\rightarrow \infty$ and $N\rightarrow \infty$ with $\frac{N}{M}=d$ by ${M \xlongrightarrow{d} \infty}$ .




\subsection{First-order Analysis}

\label{sec_pre}

\subsubsection{Deterministic Equivalent of the Normalized MI}
The first order analysis of the MI over non-centered non-separable MIMO channels can be obtained by utilizing the deterministic equivalent (DE)~\cite{hachem2007deterministic}. The DE provides an accurate large system approximation and is widely used in the analysis of large MIMO systems~\cite{hachem2007deterministic,zhang2022outage}. The parameters in the DE are determined by a system of equations. Denote $\bold{\Sigma}^{(i)}$ and $\bold{\Sigma}^{[j]}$ as the $i$-th row and $j$-th column of $\bold{\Sigma}$ and define matrices $\bold{D}_{j}\in \mathbb{R}^{M\times M}=\diag(\bold{\Sigma}^{[j]})$ and $\widetilde{\bold{D}}_{i} \in \mathbb{R}^{N\times N}=\diag(\bold{\Sigma}^{(i)})$. For the model in~(\ref{H_reform}),  we consider the following system of $N+M$ equations
\begin{equation} 
\label{sys_eq}
\begin{cases}
  &\delta_j=\frac{\Tr\BD_j \BT(z)}{M},~j=1,2,...,M,
 \\
 &\widetilde{\delta}_i=\frac{\Tr\widetilde{\BD}_i \widetilde{\BT}(z)}{M},~i=1,2,...,N.
\end{cases}
\end{equation}
with 
\begin{equation} 
\begin{aligned}
\label{tz_def}
\FT(z)&=\left( \boldsymbol{\psi}^{-1}(z)-z\BA\widetilde{\boldsymbol{\psi}}(z)\BA^{H}  \right)^{-1},
\\
\widetilde{\FT}(z)&=\left( \widetilde{\boldsymbol{\psi}}^{-1}(z)-z\BA^{H}{\boldsymbol{\psi}}(z)\BA  \right)^{-1},
\\
\boldsymbol{\psi}(z)&=\diag(\psi_{i}(z),~1\le i \le N),
\\
\widetilde{\boldsymbol{\psi}}(z)&=\diag(\widetilde{\psi}_{j}(z),~1\le j \le M),
\end{aligned}
\end{equation}
and
\begin{equation} 
\begin{aligned}
\psi_{i}(z)&=-z^{-1}(1+\widetilde{\delta}_i)^{-1}
\\
&
=-z^{-1}(1+\frac{1}{M}\Tr\widetilde{\BD}_{i}\widetilde{\FT}(z))^{-1},~1 \le i \le N,
\\
\widetilde{\psi}_{j}(z)&=-z^{-1}(1+{\delta}_j)^{-1}
\\
&=-z^{-1}(1+\frac{1}{M}\Tr{\BD}_{j}{\FT}(z))^{-1}, 1 \le j \le M.
\end{aligned}
\end{equation}
The existence and uniqueness of the solution for~(\ref{sys_eq}) have been shown in~\cite[Theorem 2.4]{hachem2007deterministic} and the system of equations can be solved by the fixed-point algorithm shown in Algorithm~\ref{sol_fund}.
\begin{algorithm} 
\caption{ Fixed-point algorithm for $\delta_j, \widetilde{\delta}_i$  } 
\label{sol_fund} 
\begin{algorithmic}[1] 
\REQUIRE  $z$, $\bold{A}$, $\bold{\Sigma}$, $\delta^{(0)}_{j}>0$, $j=1,2,...,M$, $\widetilde{\delta}^{(0)}_i>0$, $i=1,2,...,N$, and set $t=1$.
\REPEAT
\STATE $\delta^{(t)}_{j}$ $\leftarrow$ $\frac{\Tr\BD_j \BT^{(t-1)}(z)}{M}$, $j=1,2,...,M$.
\STATE $\widetilde{\delta}^{(t)}_i$ $\leftarrow$ $\frac{\Tr\widetilde{\BD}_i \widetilde{\BT}^{(t)}(z)}{M}$, $i=1,2,...,N$. 
\STATE $t \leftarrow  t+1$.
\UNTIL Convergence.
\ENSURE  $\delta_j, \widetilde{\delta}_i$.
\end{algorithmic}
\end{algorithm}
The first-order analysis of the MI can be investigated by the method in~\cite{hachem2007deterministic} and is given by the following lemma.
\begin{lemma} 
\label{mean_the}
(\hspace{-0.01cm}\cite[Theorem 4.1]{hachem2007deterministic}) 
Given assumptions~\textbf{A.1}-\textbf{A.3}, the following convergence holds true
\begin{equation} 
\begin{aligned}
\frac{1}{N} C_{M}(\zeta) \xrightarrow[N \xlongrightarrow{d} \infty]{a.s.} \frac{1}{N}\overline{C}_{M}(\zeta),
\end{aligned}
\end{equation}
where 
\begin{equation}
\label{mean_exp} 
\begin{aligned}
\overline{C}_{M}(\zeta)&=\log\det[\zeta^{-1}\FT^{-1}(-\zeta)] -\log\det [\zeta \widetilde{\boldsymbol{\psi}}(-\zeta)]
-\frac{\zeta}{M}\sum_{i,j}\sigma^2_{ij}[\FT(-\zeta)]_{i,i}[\widetilde{\FT}(-\zeta)]_{j,j} .
\end{aligned}
\end{equation}
\end{lemma}
Lemma~\ref{mean_the} gives the DE for the normalized MI with a non-separable variance profile. When $\sigma_{i,j}^2=d_i \widetilde{d}_j$ and $\BA=\bold{0}$,~(\ref{mean_exp}) degenerates to~\cite[Eq. (58)]{pizzo2022fourier} for the centered, separable case. Although Lemma~\ref{mean_the} does not prove that $\E C_{M}(\zeta)\xrightarrow {M \xlongrightarrow{d} \infty}  \overline{C}_{M}(\zeta)$, it indicates that $\overline{C}_{M}(\zeta)$ is a good approximation. In the following, we will focus on characterizing the distribution of the MI.


\subsection{Second-order Analysis}
\label{sec_main}
In this section, we will set up the CLT for non-centered non-separable correlation MIMO channels, which is the main contribution of this paper. The result will then be utilized to derive a closed-form approximation for the outage probability. For that purpose, we first introduce some notations that will be used for deriving the asymptotic variance of the MI.

\subsubsection{Notations}
In the following, we will omit $(z)$ in $\BT(z)$ and $\widetilde{\BT}(z)$ for simplicity. We first introduce four matrices $\bold{\Pi}_{i}, \bold{\Xi}_{i}, \bold{\Gamma}_{i}, \widetilde{\bold{\Lambda}} \in \mathbb{R}^{i\times i}$, $i \le M$, $j,k=1,2,...,i$, with
\begin{equation}
\label{coef_def}
\begin{aligned}
[\bold{\Pi}_{i}]_{j,k}&=\Pi_{j,k}=\frac{ \Ba^{H}_k \BT\BD_{j}\BT\Ba_{k}}{M(1+\delta_k)^{2}},
\\
[\bold{\Xi}_{i}]_{j,k}&=\Xi_{j,k}=\frac{ \Ba^{H}_k \BT\Ba_{j} \Ba_{j}^{H} \BT\Ba_{k}}{(1+\delta_{j})^{2}(1+\delta_{k})^2 }\mathbbm{1}_{k\neq j},
\\
[\bold{\Gamma}_{i}]_{j,k}&=\Gamma_{j,k}=\frac{\Tr\BD_j\BT\BD_k\BT}{M^2},
\\
\widetilde{\bold{\Lambda}}_{i}&=\rho^2\diag(\widetilde{t}_{1,1}^2,\widetilde{t}_{2,2}^2,...,\widetilde{t}_{i,i}^2),
\end{aligned}
\end{equation}
where $\delta_{i}$ and $\FT$ are given in~(\ref{sys_eq}) and~(\ref{tz_def}), respectively, $\widetilde{t}_{i,i}$ is the diagonal entry of $\widetilde{\bold{T}}$, and $\bold{a}_i$ denotes the $i$-th column of $\BA$. Further define matrix $\bold{B}_{i} \in \mathbb{R}^{2i\times 2i}$ as
\begin{equation}
\label{def_BB}
\bold{B}_{i}=
\begin{bmatrix}
\bold{\Pi}_{i} & \bold{\Gamma}_{i}
\\
\bold{\Xi}_{i}+\widetilde{\bold{\Lambda}}_{i} &\bold{\Pi}_i^{T}
\end{bmatrix},
\end{equation}
which will be used in the derivation of the asymptotic variance.


\subsubsection{Asymptotic Distribution of the MI}
The following theorem characterizes the asymptotic distribution of the MI for non-centered non-separable MIMO systems.
\begin{theorem} (The CLT for ${C}_{M}(\zeta)$) 
\label{main_clt} If assumptions~\textbf{A.1}-\textbf{A.3} hold true, the MI ${C}_{M}(\zeta)$ satisfies 
\begin{equation} 
\frac{{C}_{M}(\zeta)-\E {C}_{M}(\zeta)}{\sqrt{V_{M}(\zeta)}} \xrightarrow[M \xlongrightarrow{d} \infty]{\mathcal{D}} \mathcal{N}(0, 1),
\end{equation}
where 
\begin{equation}
\label{variance_exp}
V_{M}(\zeta)=-\log\det(\bold{I}_{M}-\bold{B}_{M}),
\end{equation}
with $\bold{B}_{M}$ defined in~(\ref{def_BB}).
\end{theorem}
\begin{proof} The proof of Theorem~\ref{main_clt} is given in Appendix~\ref{proof_main_clt}.
\end{proof}
Theorem~\ref{main_clt} indicates that the asymptotic distribution of the MI is a Gaussian distribution, whose mean and variance are given by the parameters determined from~(\ref{sys_eq}). With Lemma~\ref{mean_the} and Theorem~\ref{main_clt}, and $\zeta=\frac{\sigma^2}{G_R G_S N_R N_S}$, we can obtain the large system approximation for the outage probability of holographic MIMO systems.
\begin{proposition} \label{outage_prop}
(Outage probability of holographic MIMO systems) Given a rate threshold $R$, the outage probability $P_{out}(R)$ of the holographic system can be approximated by
\begin{equation} 
\label{outage_exp}
P_{out}(R)\approx \phi\left(\frac{R-\overline{C}_{M}(\frac{\sigma^2}{G_R G_S N_R N_S})}{\sqrt{V_M(\frac{\sigma^2}{ G_R G_S N_R N_S})} }\right).
\end{equation}
\end{proposition}



\subsubsection{Comparison with Existing Works}\hfill\\
\indent\textbf{1. Separable Correlation with LoS.} The CLT of the MI over channels with separable NLoS and LoS components can be obtained by the method from~\cite[Theorem 2.2]{hachem2012clt}, which is a special case of Theorem~\ref{main_clt}. In fact, when $\bold{\Sigma}$ is separable,~(\ref{sys_eq}) will become a system of two equations with respect to the parameters $\delta$ and $\widetilde{\delta}$, with $\delta_i={d}_{j}\delta $ and $\widetilde{\delta}_i=\widetilde{d}_{i}\widetilde{\delta} $ such that $\bold{B}_{M}$ in~(\ref{def_BB}) degenerates to a matrix in $\mathbb{R}^{2\times 2}$. Then, the result in Theorem~\ref{main_clt} will degenerate to~\cite[Theorem 2.2]{hachem2012clt}.

\textbf{2. Non-separable Correlation without LoS.} The CLT of the MI over channels with non-separable profile but without LoS components can be obtained by the method from~\cite[Theorem 3.2]{hachem2008clt}, which is also a special case of Theorem~\ref{main_clt}. Without LoS, i.e., $\bold{A}=\bold{0}$, $\bold{B}_{M}$ in~(\ref{def_BB}) degenerates to an $M$-by-$M$ matrix and Theorem~\ref{main_clt} will degenerate to~\cite[Theorem 3.2]{hachem2008clt}.

\textbf{3. Deterministic Model.} 
When $\BY=\bold{0}$, $\BH$ degenerates to the deterministic model in~\cite[Eq. (12)]{pizzo2020holographic} and there is no fluctuation for the MI. In this case, the MI in~(\ref{mi_second}) is an approximation of the MI in~\cite[Eq. (18)]{wan2023can} by using the Fourier expansion and can be obtained by computing the logarithm of the deterministic determinant.


\section{Proof of Theorem~\ref{main_clt}}
\label{proof_main_clt}
In this section, we will utilize the martingale approach to show the asymptotic Gaussianity of $C_{M}(\rho)-\E C_{M}(\rho)$. The martingale method can be traced back to Girko’s REFORM (REsolvent, FORmula and Martingale) method~\cite{girko2003thirty} and is widely used~\cite{hachem2008clt,hachem2012clt,zhang2022outage}. Specifically, we first decompose $C_{M}(\rho)-\E C_{M}(\rho)$ into a sum of martingale differences. Based on the CLT for the martingale (Lemma~\ref{clt_lemma}), we verify the CLT conditions (Section~\ref{lyn_veri}). Then, in the process of deriving the closed-form asymptotic variance, we set up a system of equations to compute the intermediate quantities based on the resolvent evaluation (Section~\ref{eva_var}). To approximate the random quantities, a large amount of involved computation relies on the quadratic form of random vectors (Lemma~\ref{lemma_ext}).

The key steps of proving Theorem~\ref{main_clt} follow from the CLT for the martingales given in Lemma~\ref{clt_lemma}.
\begin{lemma}\label{clt_lemma} (CLT for the martingale~\cite[Theorem 35.12]{billingsley2008probability}) \label{clt_lemma} Let $\gamma_{1}^{(n)}$, $\gamma_{2}^{(n)}$, ..., $\gamma_{n}^{(n)}$ be a sequence of martingale differences with respect to the increasing filtration $\mathcal{F}_{1}^{(n)}$, $\mathcal{F}_{2}^{(n)}$, ..., $\mathcal{F}_{n}^{(n)}$. Assume that there exists a sequence of real positive numbers $s_{n}^2$ with $\lim \inf\limits_{n} s_{n}^2 > 0$ such that 
the Lyapunov's condition holds true
\begin{equation} 
\label{lyya_condi}
 \exists \delta >0, ~\sum_{i=1}^{n}\E |\gamma_{i}^{(n)}|^{2+\delta}  \xrightarrow{n \rightarrow \infty} 0,
\end{equation}
 and 
\begin{equation} 
\label{var_mart_dep}
\sum_{i=1}^{n}\E_{j-1} [ \gamma_{j}^{(n)}   ]^2 -s_{n}^2 \xrightarrow{n \rightarrow \infty} 0.
\end{equation}
Then, $\sum_{j=1}^{n}\frac{\gamma_{j}^{(n)}}{s_n} $ converges to $\mathcal{N}(0,1)$.
\end{lemma}
In fact, $s_n^2$ is the asymptotic variance of $\sum_{j=1}^{n}\gamma_{j}^{(n)}$, which can be computed by evaluating the sum of martingale differences in~(\ref{var_mart_dep}). Lemma~\ref{clt_lemma} indicates that there are three main steps to prove Theorem~\ref{main_clt}: 

1. Validate the Lyapunov's condition in~(\ref{lyya_condi}). 

2. Compute the asymptotic variance by~(\ref{var_mart_dep}). 

3. Validate $\lim \inf\limits_{n} s_{n}^2 > 0$.

Before we proceed, we first introduce the resolvent matrices of $\BH\BH^{H}$ and $\BH^{H}\BH$, given by
\begin{equation}
\begin{aligned}
\BQ(z)&=\left(-z\bold{I}_{N}+\BH\BH^{H} \right)^{-1},
\\
\widetilde{\BQ}(z)&=\left(-z\bold{I}_{M}+\BH^{H}\BH \right)^{-1}.
\end{aligned}
\end{equation}
$\BQ_j(z)$ represents the rank-one perturbations of $\BQ(z)$, given by
\begin{equation}
\begin{aligned}
\BQ_j(z)&=\left(-z\bold{I}_{N}+\BH_j\BH_j^{H} \right)^{-1},
\end{aligned}
\end{equation}
where $\BH_j$ is obtained by removing the $j$-th column from $\BH$. The diagonal entry of $\widetilde{Q}(z)$ can be obtained as
\begin{equation}
\label{qii_exp}
\widetilde{q}_{i,i}=\frac{1}{-z(1+\Bh_i^{H}\BQ_i\Bh_i)}.
\end{equation}
In the following, we will omit $(z)$ in $\BQ(z)$ and use the notation $\rho=-z$.

Define the $\sigma$-filed $\mathcal{F}_{j}=\sigma(\bold{y}_{l}, j \le l \le M )$ generated by $\bold{y}_{j}$, $\bold{y}_{j+1}$,..., $\bold{y}_{M}$ and denote the conditional expectation with respect to $\mathcal{F}_{j}$ by $\E_{j}=\E (\cdot|\mathcal{F}_{j})$. $C_{M}(\rho)-\E C_{M}(\rho)$ can be decomposed into a sum of martingale differences as follows
\begin{equation} 
\label{martin_decomp}
\begin{aligned}
&\log\det(\BH\BH^{H}+\rho\bold{I}_{N})-\E\log\det(\BH\BH^{H}+\rho\bold{I}_{N})
\\
&\overset{(a)}{=}-\sum_{j=1}^{M} (\E_{j}-\E_{j+1})\log\left(\frac{\det(\BH_{j}\BH_{j}^{H}+\rho\bold{I}_{N})}{\det(\BH\BH^{H}+\rho\bold{I}_{N})} \right) 
\\
&\overset{(b)}{=}-\sum_{j=1}^{M} (\E_{j}-\E_{j+1})\log\left(\frac{\det(\BH_{j}^{H}\BH_{j}+\rho\bold{I}_{M-1})}{\det(\BH^{H}\BH+\rho\bold{I}_{M-1})} \right) 
\\
&\overset{(c)}{=}-\sum_{j=1}^{M} (\E_{j}-\E_{j+1})\log[\widetilde{\bold{Q}}]_{j,j}
\\
&\overset{(d)}{=}-\sum_{j=1}^{M} (\E_{j}-\E_{j+1})\log(1+\Bh_{j}\BQ_{j}\Bh_{j})
\\
&\overset{(e)}{=}-\sum_{j=1}^{M} (\E_{j}-\E_{j+1})\log(1+\zeta_{j})
\\
&=-\sum_{j=1}^{M}\gamma_{j}
,
\end{aligned}
\end{equation}
where $\zeta_{j}=\frac{\Bh_{j}\BQ_{j}\Bh_{j}-\frac{1}{M}\Tr\BD_j\BQ_j-\Ba_{j}^{H}\BQ_j\Ba_j}{1+\frac{1}{M}\Tr\BD_j\BQ_j+\Ba_{j}^{H}\BQ_j\Ba_j}$ and $\gamma_{j}= (\E_{j}-\E_{j+1})\log(1+\zeta_{j})$. Step $(a)$ in~(\ref{martin_decomp}) follows from $\E_{j}\log\det(\BH_{j}\BH_{j}^{H}+\rho\bold{I}_{N})=\E_{j+1}\log\det(\BH_{j}\BH_{j}^{H}+\rho\bold{I}_{N})$, step $(b)$ follows from the identity $\log\det(\bold{I}+\bold{A}\bold{B})=\log\det(\bold{I}+\bold{B}\bold{A})$, steps $(c)$ and $(d)$ can be obtained by using~the matrix inversion formula and~(\ref{qii_exp}), respectively, and step $(e)$ is derived by adding $-\log(1+\frac{1}{M}\Tr\BD_j\BQ_j+\Ba_{j}^{H}\BQ_j\Ba_j)$ to $\E_j$ and $\E_{j+1}$. By far, we have decomposed $C_{M}(\rho)-\E C_{M}(\rho) $ into the summation of the sequence $\gamma_{M}$, $\gamma_{M-1}$, ..., $\gamma_1$, which is a martingale difference with respect to the increasing filtration $\mathcal{F}_{M}$, $\mathcal{F}_{M-1}$,..., $\mathcal{F}_{1}$. Next, we will finish the three steps of the proof.

\subsection{Step 1: Validation of Lyapunov's condition}
\label{lyn_veri}
By~(\ref{martin_decomp}), we need to validate $\sum_{j=1}^{M}\E|\gamma_{j}|^{2+\delta} \xrightarrow{M \rightarrow \infty} 0$, where the left hand side can be evaluated as,
\begin{equation} 
\begin{aligned}
\label{lyn_bnd}
&(\E|\gamma_{j}|^{2+\delta} )^{\frac{1}{2+\delta}} \le (\E |\E_j \log(1+\zeta_j)   |^{2+\delta}  )^{\frac{1}{2+\delta}}
+(\E |\E_{j+1} \log(1+\zeta_j)   |^{2+\delta}  )^{\frac{1}{2+\delta}}
\\
&
\overset{(a)}{\le} 2 (\E | \log(1+\zeta_j)   |^{2+\delta}  )^{\frac{1}{2+\delta}}.
\end{aligned}
\end{equation}
The inequality $(a)$ in~(\ref{lyn_bnd}) follows from $(|\E_j\log(1+\zeta_j)|)^{2+\delta} \le  (\E_j| \log(1+\zeta_j)|)^{2+\delta}    \le \E_{j} ( |\log(1+\zeta_j)|)^{2+\delta}$ according to Jensen's inequality. By Assumptions~\textbf{A.1}-\textbf{A.3} and the trace inequality in~(\ref{trace_inq}), it follows
\begin{equation} 
\frac{1}{M}\Tr\BD_j \BQ_j +\Ba_{j}^{H}\BQ_j\Ba_j \le \rho^{-1}( c \sigma^2_{max}+ a_{max}^2):=U_1.
\end{equation}
By the non-decreasing attribute of $f(x)=\frac{x}{1+x}$,  we have
\begin{equation} 
\label{zeta_bnd}
\begin{aligned}
&\zeta_j \ge - \frac{\frac{1}{M}\Tr\BD_j\BQ_j+\Ba_{j}^{H}\BQ_j\Ba_j}{1+\frac{1}{M}\Tr\BD_j\BQ_j+\Ba_{j}^{H}\BQ_j\Ba_j} 
\ge  -\frac{U_1}{1+U_1}:=-U_2 >-1.
\end{aligned}
\end{equation}
By the non-decreasing attribute of $g(x)=\frac{\log(1+x)}{x},~x\in (-1,\infty)$, we have
\begin{equation} 
\label{zeta_bnd2}
\frac{\gamma_j}{\zeta_j} \le \frac{\log(1-U_2)}{U_2}:=U_3 .
\end{equation}
Therefore, according to~(\ref{lyn_bnd}),~(\ref{zeta_bnd}), and~(\ref{zeta_bnd2}), if $\delta \in (0, 6]$ and $\E |X_{i,j}|^{2+\delta}\le \infty$, we have 
\begin{align}
\label{lya_vad}
&\E |\gamma_j|^{2+\delta} \le (2 U_3)^{2+\delta} \E |\zeta_j|^{2+\delta} \le (2 U_3)^{2+\delta} 
\E |  \Bh_{j}^{H}\BQ_{j}\Bh_{j}-\frac{1}{M}\Tr\BD\BQ_j-\Ba_{j}^{H}\BQ_j\Ba_j|^{2+\delta}
\overset{(a)}{\le} \frac{U}{M^{\frac{2+\delta}{2}}},
\end{align}
where $(a)$ in~(\ref{lya_vad}) follows from~Lemma~\ref{lemma_ext} in Appendix~\ref{appendix_lemmas}. Therefore, $\sum_{j=1}^{M}\E|\gamma_{j}|^{2+\delta} \xrightarrow{M \rightarrow \infty} 0$ and the Lyapunov's condition is validated.

\subsection{Step 2: Evaluation of the asymptotic variance}
In this step, we will first derive the asymptotic variance by evaluating the sum of the mean square for the martingale difference.

\subsubsection{Evaluation of the sum of conditional variances}
\label{eva_var}
In the following, we will show the following convergence
\begin{equation} 
\label{first_app_var}
\sum_{j=1}^{M}\E_{j+1}\gamma_j^2 \xrightarrow[M\rightarrow \infty]{\mathcal{P}}\sum_{j=1}^{M}\E_{j+1}(\E_j \zeta_j)^2.
\end{equation}
This convergence intuitively follows from the first-order Taylor expansion
\begin{equation}
\begin{aligned}
&\E_j \log(1+\zeta_j)=\E_j \zeta_j 
+[\E_j \log(1+\zeta_j)\mathbbm{1}_{|\zeta_j|\le U_2}-\E_j \zeta_j]
+ \E_j\log(1+\zeta_j)
=\E_j \zeta_j 
+\varepsilon_{j,1}+\varepsilon_{j,2},
\end{aligned}
\end{equation}
where $U_2$ is given in~(\ref{zeta_bnd}). Now we will show that $\sum_{j=1}^{M}\varepsilon_{j,1}$ and $\sum_{j=1}^{M}\varepsilon_{j,2}$ vanish as $M$ approaches infinity. By the Taylor expansion of $\log(1+x)$, we have
\begin{equation}
\begin{aligned}
&
|\varepsilon_{j,1}|=| \E_j (\sum_{m=1}^{\infty} \frac{\zeta_j \mathbbm{1}_{|\zeta_j |<U_2}}{m }  - \zeta_j )|
\le \E_j \zeta_j \mathbbm{1}_{\zeta_j >U_2} 
+
\sum_{m=2}^{\infty}\E_j |\zeta_j|^{m}\mathbbm{1}_{|\zeta_j| \le U_2}
\le  \E_j \zeta_j \mathbbm{1}_{\zeta_j >U_2} + \frac{\E_j \zeta_j^2\mathbbm{1}_{|\zeta_j | \le U_2}  }{1-U_2}.
\end{aligned}
\end{equation}
By the inequality $(a+b)^2\le 2(a^2+b^2)$, we can obtain
\begin{equation}
\begin{aligned}
\label{varep_1}
\E \varepsilon_{j,1}^2 & \le 2 \left(\E \zeta_j^2 \mathbbm{1}_{\zeta_j >U_2}+\frac{ \E \zeta_j^4\mathbbm{1}_{|\zeta_j | \le U_2}  }{1-U_2}   \right)
\\
&
\overset{(a)}{\le } 2(\frac{\E \zeta_j^4}{U_2^2}+\frac{\E \zeta_j^4}{(1-U_2)^2})
\le 2(\frac{1}{U_2^2}+\frac{1}{(1-U_2)^2})
\E (\Bh_i^{H}\BQ_j\Bh_i-\frac{1}{M}\Tr\BD_j\BQ_j -\Ba_i^{H}\BQ_j\Ba_i)^4 \overset{(b)}{\le} K M^{-2},
\end{aligned}
\end{equation}
where the inequality $(a)$ in~(\ref{varep_1}) follows from $\zeta_j^2 \mathbbm{1}_{\zeta_j >U_2}  \le   \zeta_j^2 \frac{\zeta_j^2 }{U_2} \mathbbm{1}_{\zeta_j >U_2} $ and step (b) in~(\ref{varep_1}) follows from~(\ref{qua_ext}) in Lemma~\ref{lemma_ext} when $\E |X_{i,j}|^{8} \le \infty$, which holds true for Gaussian entries. $\E \varepsilon_{j,2}^2$ can be handled similarly as
\begin{equation}
\E \varepsilon_{j,2}^2 \le \E \zeta_j^2 \mathbbm{1}_{\zeta_j >U_2} \le U_2^{-2} \E \zeta_j^4 \mathbbm{1}_{\zeta_j >U_2} \le K' M^{-2}.
\end{equation}
By far, we have obtained $\E_j \log(1+\zeta_j)=\E_j \zeta_j +\BO(M^{-2}) $ and we can similarly obtain $\E_{j+1} \log(1+\zeta_j)=\E_j+1 \zeta_j +\BO(M^{-2}) $. Since $\E_{j+1}\zeta_j=0 $, we have $\gamma_j=\E_j\zeta_j+\varepsilon_{z,j}$ and $\E \varepsilon_{z,j}^2=\BO(M^{-2})$ to further obtain
\begin{equation}
\begin{aligned}
&
\E| \E_{j+1}(\gamma_j)^2-\E_{j+1} (\E_j\zeta_j)^2|
=\E|\E_{j+1}\varepsilon_{z,j}^2+2\E_{j+1} |\varepsilon_{z,j}\E_j\zeta_j||
\\
&
\le \E \varepsilon_{z,j}^2+ 2\E^{\frac{1}{2}} \varepsilon_{z,j}^2 \E^{\frac{1}{2}}|\E_j\zeta_j|^2 
\le \E \varepsilon_{z,j}^2+ 2\E^{\frac{1}{2}} \varepsilon_{z,j}^2 \E^{\frac{1}{2}}\zeta_j^4  
=\BO(M^{-\frac{3}{2}}).
\end{aligned}
\end{equation}
Therefore, we have 
\begin{equation}
\sum_{j=1}^{M}\E | \E_{j+1}\gamma_j^2-\E_{j+1}(\E_j\zeta_j)^2 |\le K_{c} M^{-\frac{1}{2}}.
\end{equation}
By the Markov's inequality, we can conclude~(\ref{first_app_var}). Therefore, the variance evaluation turns to be the evaluation of $\sum_{j=1}^{M}\E_{j+1}(\E_j \zeta_j)^2$, which will be handled in the following.

\subsubsection{Evaluation of $\sum_{j=1}^{M}\E_{j+1}(\E_j \zeta_j)^2$}
Before we proceed, we introduce
\begin{equation}
\widetilde{b}_{i}=\rho^{-1}\left(1+\Ba_{i}^{H}\BQ_{i}\Ba_{i}+\frac{1}{M}\Tr\BD_{i}\BQ_{i} \right)^{-1},
\end{equation}
which can be regarded as an intermediate approximation for $\widetilde{q}_{i,i}$. Then, we have $\E_{j}\zeta_j=\E \rho \widetilde{b}_{j} e_j$, where
\begin{equation}
e_j=\Bh_j^{H}\BQ_j\Bh_j-\Ba_{j}^{H}\BQ_{j}\Ba_{j}-\frac{1}{M}\Tr\BD_{j}\BQ_{j}.
\end{equation}
According to~(\ref{qjj_approx}) in Lemma~\ref{appro_diag_ele} and~(\ref{qua_ext}) in Lemma~(\ref{lemma_ext}), we can obtain $\E |\widetilde{b}_{j}-\widetilde{t}_{j,j}|\le 2 (\E | \widetilde{b}_{j}-\widetilde{q}_{j,j} |^2 +\E | \widetilde{t}_{j,j}-\widetilde{q}_{j,j} |^2) =o(1) $ so that
\begin{equation}
\begin{aligned}
\E |  \E_{j-1} (\E_{j}\rho  \widetilde{b}_{j} e_j  )^2 - \E_{j-1} \rho^2\widetilde{t}_{j,j}^2 (\E_{j-1} e_j)^2 | 
&\le \E \rho^2|\E_{j}(\widetilde{b}_{j}-\widetilde{t}_{j,j})e_{j}(\E_{j}(\widetilde{b}_{j}+\widetilde{t}_{j,j})e_{j}) |
\\
&
 \le K \E^{\frac{1}{2}}|\widetilde{b}_{j}-\widetilde{t}_{j,j}  |^2   \E^{\frac{1}{2}} e_{j}^4=o(M^{-1}).
\end{aligned}
\end{equation}
Therefore, we can obtain $\E \sum_{j=1}^{M} \E_{j-1}(\E_{j}\rho\widetilde{b}_{j}e_{j} )^2 -\E \sum_{j=1}^{M} \E_{j-1}(\E_{j}\rho\widetilde{t}_{,jj}e_{j} )^2\xrightarrow{M \rightarrow \infty} 0 $ and conclude
\begin{equation}
 \sum_{j=1}^{M} \E_{j-1}(\E_{j}\rho\widetilde{b}_{j}e_{j} )^2\xrightarrow[M \rightarrow \infty]{\mathcal{P}}  \sum_{j=1}^{M} \E_{j-1}(\E_{j}\rho\widetilde{t}_{j,j}e_{j} )^2,
\end{equation}
by the Markov's inequality. Then, the evaluation of the asymptotic variance resorts to the evaluation of $\sum_{j=1}^{M}\rho^2\widetilde{t}_{j,j}^2 \E_{j-1} (\E_{j}e_{j})^2$. 
\begin{figure*}
\begin{equation}
\label{sum_asymp_var}
\begin{aligned}
&\sum_{j=1}^{M}\rho^2\widetilde{t}_{j,j}^2 \E_{j-1} (\E_{j}e_{j})^2
=\sum_{j=1}^{M} \rho^2\widetilde{t}_{j,j}^2 \left( \frac{\Tr\BD_{j}(\E_j \BQ_j) \BD_{j}(\E_j \BQ_j)}{M^2}+\frac{2\Ba_j^{H}(\E_j \BQ_j)\BD_j \Ba_j^{H}(\E_j \BQ_j) }{M}  \right)
:=T.
\end{aligned}
\end{equation}
\hrulefill
\end{figure*}

By similar lines as in~\cite{hachem2012clt}, it can be proved that $\Var(\Tr\BD_{j}(\E_j \BQ_j) \BD_{j}(\E_j \BQ_j))=\BO(1)$ and $\Var(\Ba_j^{H}(\E_j \BQ_j)\BD_j \Ba_j^{H}(\E_j \BQ_j))=\BO(M^{-1})$. Therefore, by the Chebyshev's inequality, we have the following approximation
\begin{equation}
T \xlongrightarrow[M\rightarrow \infty]{\mathcal{P}}  \sum_{j=1}^{M} (\frac{\rho^2\widetilde{t}_{j,j}^2 \psi_{j,j}}{M}+\frac{2\Theta_{j,j}}{M} ):=K_{M},
\end{equation}
where $ \psi_{j,j}=\frac{\E\Tr\BD_{j}(\E_j \BQ_j) \BD_{j} \BQ_j}{M}$ and $\Theta_{j,j}=\frac{\rho^2\widetilde{t}_{j,j}^2\E \Ba_j^{H}(\E_j \BQ_j)\BD_j \Ba_j^{H} \BQ_j}{M}$. $K_M$ can be evaluated by the following lemma.
\begin{lemma}
\label{sum_mart}
Given Assumptions~\textbf{A.1}-\textbf{A.3}, the following evaluation holds true,
\begin{equation} 
\label{sum_pjj}
 \sum_{j=1}^{M} (\frac{\rho^2\widetilde{t}_{j,j}^2 \psi_{j,j}}{M}+\frac{2\Theta_{j,j}}{M} ) \xlongrightarrow[]{M\xrightarrow{} \infty} -\log\det(\bold{I}_{2M}-\BB_M) ,
\end{equation}
where $\BB_M$ is defined in~(\ref{def_BB}).
\end{lemma}
\begin{proof}
The proof of Lemma~\ref{sum_mart} is given in Appendix~\ref{proof_sum_mart}.
\end{proof}
By~(\ref{sum_asymp_var}) and Lemma~\ref{sum_mart}, we have 
\begin{equation}
\sum_{j=1}^{M}\rho^2\widetilde{t}_{j,j}^2 \E_{j-1} (\E_{j}e_{j})^2 \xlongrightarrow[M\rightarrow \infty]{\mathcal{P}} -\log\det(\bold{I}_{2M}-\BB_M),
\end{equation}
which indicates that the asymptotic variance is $V_{M}(\sigma^2)=-\log\det(\bold{I}_{2M}-\BB_M)$. By far, we have finished Step 2.

\subsection{Step 3: The lower bound of the asymptotic variance}
In this section, we will verify $\inf\limits_{M} V_{M}(\sigma^2)>0$. By Lemma~\ref{sum_mart}, we can obtain that for large $M$, there holds true
\begin{equation} 
\label{V_M_lower}
\begin{aligned}
&
V_M(\sigma^2) -\sum_{j=1}^{M}\frac{\rho^2\widetilde{t}_{j,j}^2\sum_{i=1}^{N}\sigma_{ij}^2t_{i}^2}{M^2} 
\\
&
\ge \sum_{j=1}^{M}\left( \frac{\rho^2\widetilde{t}_{j,j}^2\Tr\BD_{j}(\E_j \BQ_j)\BD_j\BQ_j}{M^2}-\frac{\rho^2\widetilde{t}_{j,j}^2\sum_{i=1}^{N}\sigma_{ij}^2t_{i,i}^2}{M^2}\right)
\\
&
\ge \frac{\rho^2\widetilde{t}_{j,j}^2\Tr\BD_{j}(\E_j \BQ_j)\BD_j\BQ_j}{M^2}-\frac{\rho^2\widetilde{t}_{j,j}^2}{M}\frac{\Tr\BD_j\BT\BD_j\BT}{M}
\xlongrightarrow[]{M \xlongrightarrow{d}\infty} \sum_{i=1}^{M} \frac{\rho^2 \widetilde{t}_{j,j}^2}{M}[[\left(\bold{I}-\bold{B}_{j} \right)^{-1}-\bold{I}_{2j}]\bold{\Gamma}_{j}^{[j]}  ]_{j}
\\
&
=\sum_{j=1}^{M} \frac{\rho^2\widetilde{t}_{j,j}^2}{M}[\left(\bold{I}-\bold{B}_{j} \right)^{-1}\bold{B}_{j} \bold{\Gamma}_{j}^{[j]} ]_{j}\overset{(a)}{>}0,
\end{aligned}
\end{equation}
where Step (a) in (\ref{V_M_lower}) follows from the fact that entries of $\left(\bold{I}-\bold{B}_{i} \right)^{-1}$ and $\bold{B}_{i}$ are positive, which is given in the analysis above~(\ref{G_other}) in Appendix~\ref{proof_sum_mart}.

By the trace inequality in~(\ref{trace_inq}), we can obtain $\delta_{i} < c \sigma^2_{max} \rho^{-1}$ and $\widetilde{\delta}_{j} <  \sigma^2_{max} \rho^{-1}$ so that
\begin{equation}
\widetilde{t}_{j,j} > \frac{1}{\rho( 1+ c  \sigma^2_{max} \rho^{-1}+ a_{max}^2 \rho^{-1}  )}
\end{equation}
is bounded away from $0$, which also holds true for $t_{i,i}$. Therefore, we have 
\begin{equation} 
\label{V_M_last}
V_M(\rho)> \sum_{j=1}^{M} \frac{\rho^2\widetilde{t}_{j,j}^2\sum_{i=1}^{N}\sigma_{ij}^2t_{i}^2}{M^2}  >  \frac{K \sum_{i,j} \sigma_{i,j}^2 }{M^2}.
\end{equation}
By~\textbf{A.2}, we can obtain that the right hand side of~(\ref{V_M_last}) is bounded away from $0$ to conclude step 3. According to Lemma~\ref{clt_lemma}, we conclude Theorem~\ref{main_clt}.

\section{Numerical Results}
\label{sec_simu}
In this section, we will validate the theoretical results by numerical simulations. Given the equivalence of the Weichselberger and Holographic MIMO channels, we only focus on the latter. In particular, we consider the Fourier based holographic channel~\cite{pizzo2022fourier}.

\subsection{Simulation Settings}
Given the equivalence of the MI in the spatial~(\ref{mi_first}) and angular domain~(\ref{mi_second}), we only consider the channel in the angular domain, i.e., $\BH_{a}$. For the NLoS component with separable model, we consider the isotropic model, i.e., $S(k_x,k_y,\kappa_{x},\kappa_y)=1$. The variance profile for the separable case can be computed by~\cite[Eq. (70)]{marzetta2018spatially} and generated by the code in~\cite{gitluca}. Since there is no existing models for $S(k_x,k_y,\kappa_{x},\kappa_y)$ of the non-separable case, we construct one based on the product of the Gaussian kernel and the separable variance profile. The non-separable variance profile is given by
\begin{equation}
\label{nonsep_def}
\begin{aligned}
&\sigma^2(l_x,l_y,m_x,m_y)=\sigma^2_{R}(l_x,l_y)\sigma^2_{S}(m_x,m_y)
 \exp(- \frac{(l_x-m_x)^2+(l_y-m_y)^2}{a} ),
\end{aligned}
\end{equation}
where $a$ is a scaling factor and set as $a=1$. $\sigma^2_{R}(l_x,l_y)$ and $\sigma^2_{S}(m_x,m_y)$ can be computed by~\cite[Eq. (70)]{pizzo2020spatially}\footnote{The code is available at: https://github.com/lucasanguinetti/Holographic-MIMO-Small-Scale-Fading.}.
The holographic MIMO channel is generated by
\begin{equation}
\BH_{h}=\sqrt{\frac{k}{n_s}}\BA_h+\bold{\Sigma}^{\circ\frac{1}{2}}_h \odot \bold{X}_h, 
\end{equation}
where $\BX_h \in \mathbb{C}^{n_R\times n_S}$ is an i.i.d. random matrix with entries $[\BX_h]_{i,j} \sim \mathcal{CN}(0,\frac{1}{n_S})$. The LoS component $\BA_h$ is obtained by~\cite[Eq. (12)]{pizzo2020holographic} and we introduce a factor $k$ to indicate the power ratio of the LoS and non-LoS component and set $k=10$ by default. The frequency is $3\times 10^{10}$ Hz with wavelength $\lambda=0.01$ m. The SNR $\frac{1}{\sigma^2}$ is set as $10$ dB and spacing is set to be $\frac{\lambda}{4}$. The antenna efficiency is set to be $\tau=0.6$ and the antenna area is $l_a\times l_a= \frac{\lambda^2}{64}$ with $l_a=\frac{\lambda}{8} $.

\subsection{Gaussianity}
In Fig.~\ref{qqplots}, the normal quantile-quantile-plots (QQ-plots) for the normalized MI, $\frac{{C}_{M}(\zeta)- \overline{{C}}_{M}(\zeta)}{\sqrt{V_{M}(\zeta)}}$ with $\zeta=\frac{\sigma^2}{G_{R}G_{S} N_R N_S}$, are plotted by blue plus sign with settings $n_R=n_S=\{4,16,36,60 \}$.  The number of samples of the MI is $10^5$ and the red line is the quantile of the Gaussian distribution. It can be observed that the QQ plot of the normalized MI is closer to that of the QQ plot of the normalized MI as $n_R$ increases, which validates the Gaussianity of the MI.

\begin{figure*}[!htb] 
    \centering
  \subfloat[\label{qqplot4}]{
       \includegraphics[width=0.32\linewidth]{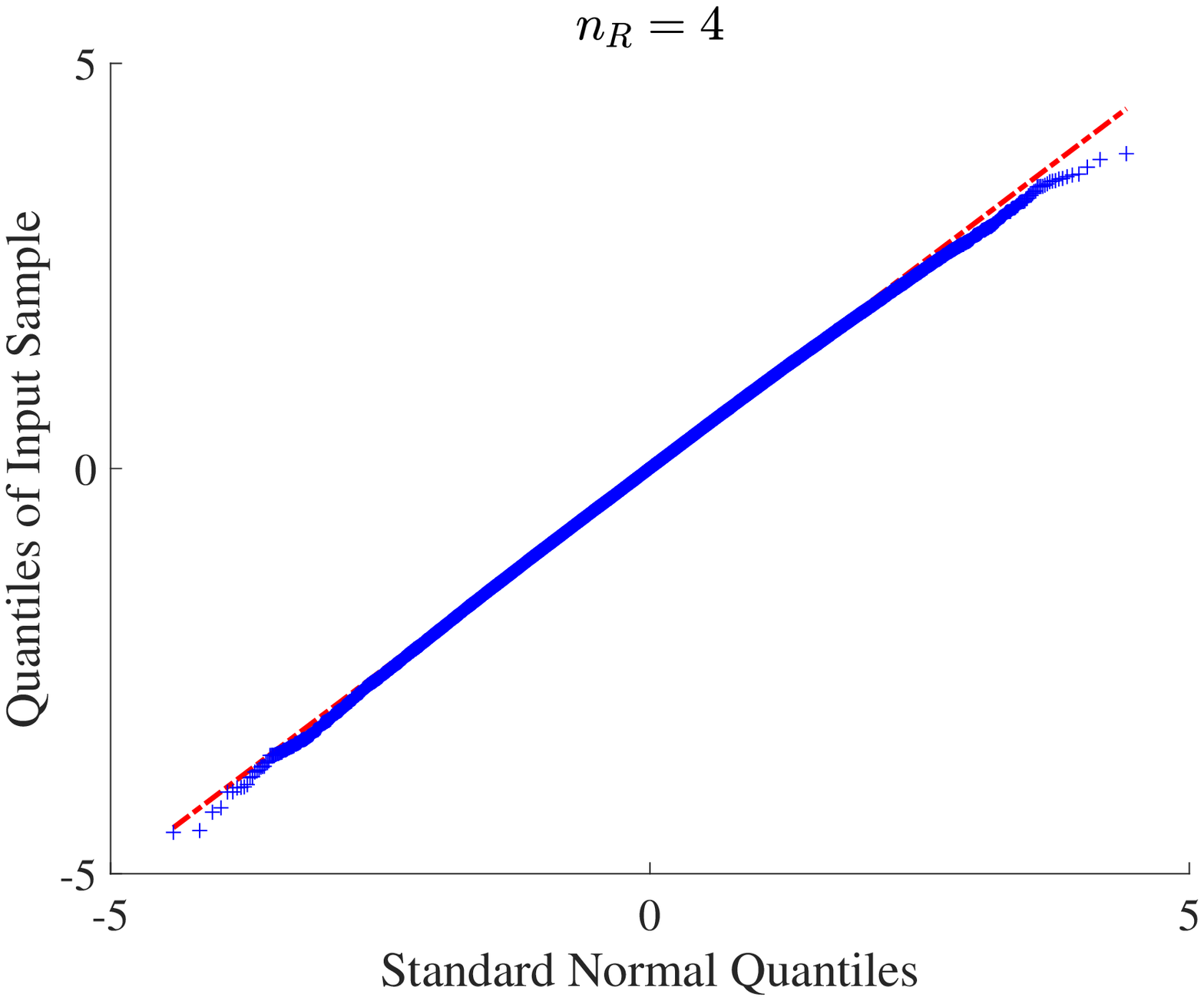}}
  \subfloat[\label{qqplot16}]{
        \includegraphics[width=0.32\linewidth]{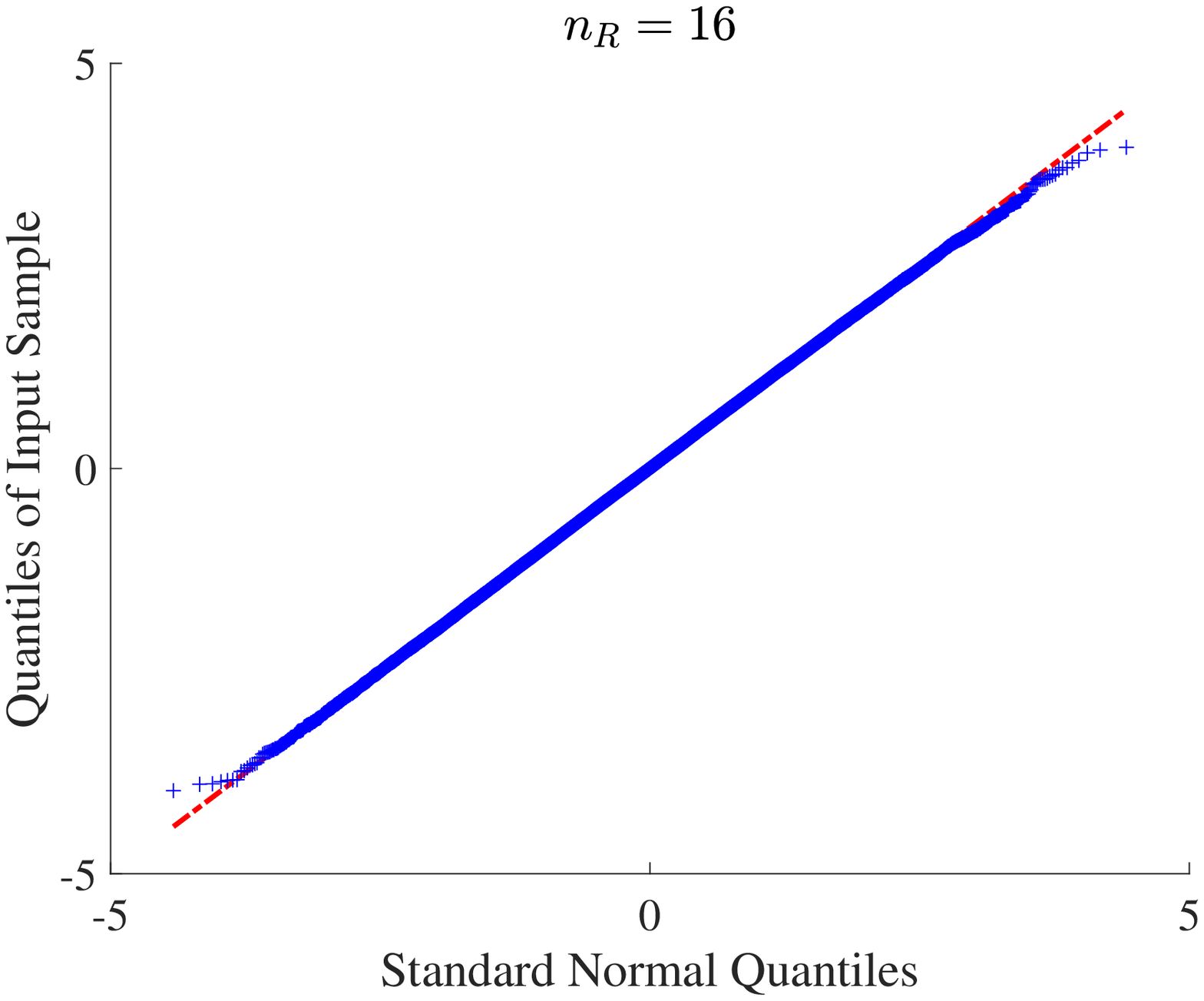}}
  \subfloat[\label{qqplot32}]{
        \includegraphics[width=0.32\linewidth]{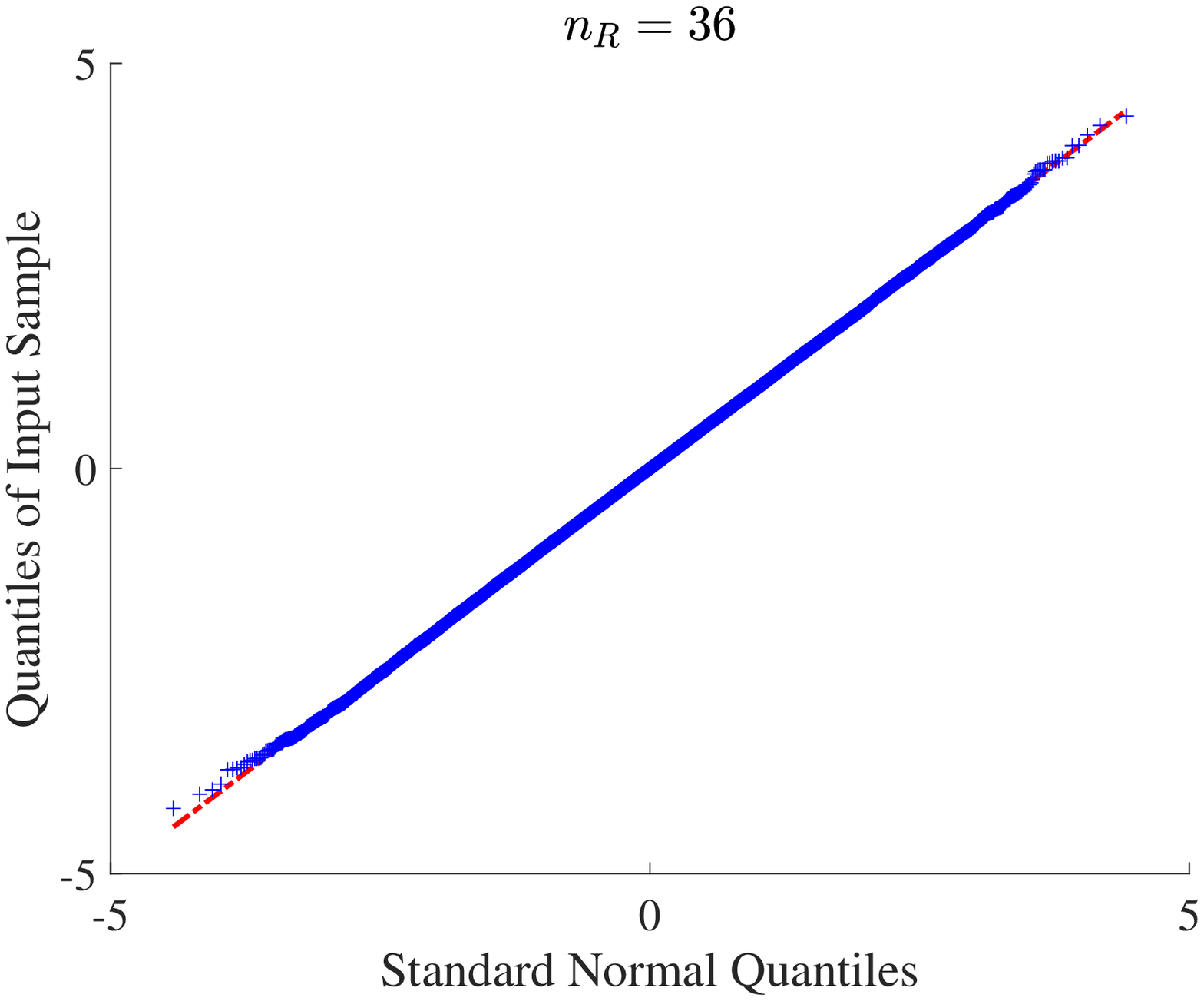}}
  \caption{Validation of Gaussianity for the normalized MI $\frac{{C}_{M}(\zeta)- \overline{{C}}_{M}(\zeta)}{\sqrt{V_{M}(\zeta)}}$.}
 \label{qqplots}
\end{figure*}

\subsection{Accuracy of the Theoretical Analysis}
The mean and variance of $C_{M}(\zeta)$ are plotted in Figs.~\ref{fig_EMI} and~\ref{fig_variance}, respectively, where the analytical values (Ana.) of the mean and variance are obtained by~(\ref{mean_exp}) in Lemma~\ref{mean_the} and~(\ref{variance_exp}) in Theorem~\ref{main_clt}, respectively. The simulation values (Sim.) are obtained by $10^{4}$ samples of $C_{M}(\zeta)$ with $\frac{1}{\sigma^2}=30$ dB. It can be observed from Figs.~\ref{fig_EMI} and~\ref{fig_variance} that the results for the mean and variance are accurate for different SNRs.

The outage probability for $\Delta=\frac{\lambda}{4}$, $L_x=L_y=10\lambda$ with $\frac{1}{\sigma^2}=\{30, 31\}$ dB is plotted in Figs.~\ref{fig_out_4}, where the analytical values are computed by~(\ref{outage_exp}) in Proposition~\ref{outage_prop} and the simulation values are generated by $5\times 10^{6}$ realizations. It can be observed from Figs.~\ref{fig_out_4} that the analytical results match the simulation values well, which validates the accuracy of the proposed outage probability approximation.


\begin{figure}[t!]
\centering\includegraphics[width=0.5\textwidth]{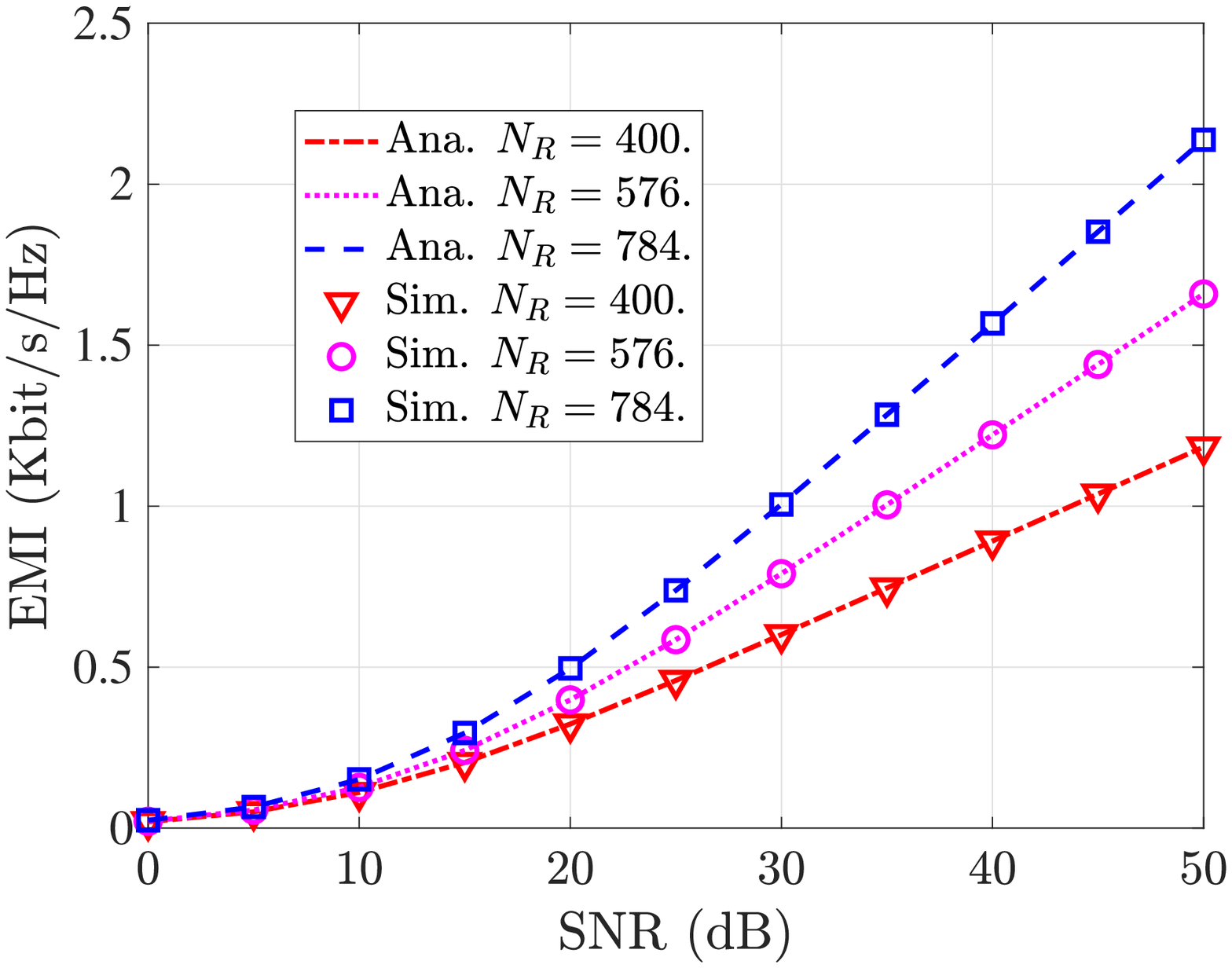}
\caption{EMI}
\label{fig_EMI}
\vspace{-0.3cm}
\end{figure}

\begin{figure}[t!]
\centering\includegraphics[width=0.5\textwidth]{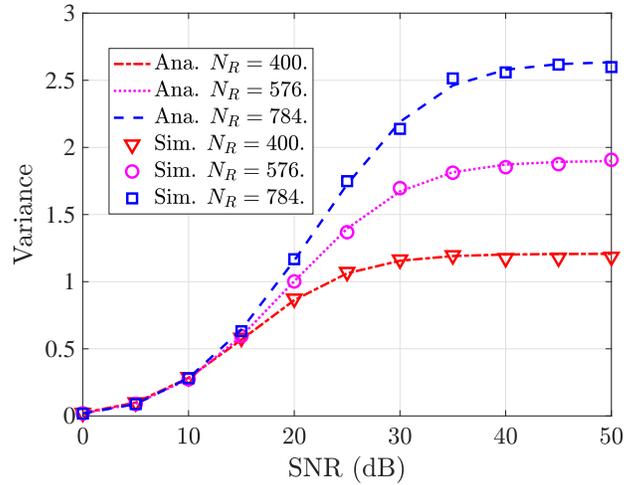}
\caption{Variance}
\label{fig_variance}
\vspace{-0.3cm}
\end{figure}

\begin{figure}[t!]
\centering\includegraphics[width=0.5\textwidth]{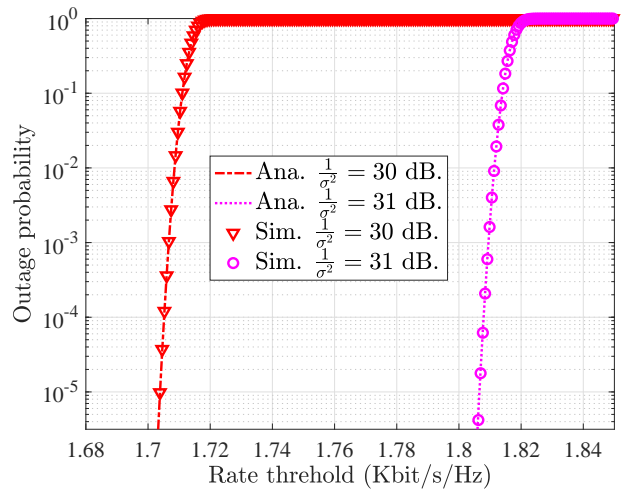}
\caption{Outage probability with $\Delta=\frac{\lambda}{4}$}
\label{fig_out_4}
\vspace{-0.3cm}
\end{figure}

\subsection{Impact of Non-Separable Correlation}
Now we compare the EMI with the non-separable correlation structure in~(\ref{nonsep_def}) and that with the separable correlation structure
\begin{equation}
\label{sep_def}
\begin{aligned}
&\sigma^2_{sep}(l_x,l_y,m_x,m_y)=m\sigma^2_{R}(l_x,l_y)\sigma^2_{S}(m_x,m_y),
\end{aligned}
\end{equation}
where $m$ is a scaling factor. For the sake of fairness,  we keep $\E \Tr\BH_h\BH_h^{H}$ the same for both cases, which is implemented by sharing the common $\BA_h$ and guaranteeing $\Tr(\bold{\Sigma}_{non-sep})= \Tr(\bold{\Sigma}_{sep})$ with appropriate $m$. The settings are $k=0.2$, $a=0.5$, and $L_x=5\lambda$. Fig.~\ref{fig_correlation_EMI} and~\ref{fig_correlation_out} depict the EMI and outage probability with separable and non-separable correlation structure, respectively. It can be observed from Fig.~\ref{fig_correlation_EMI} that the EMI of separable case is smaller than that of non-separable case, which coincides with that in~\cite{weichselberger2006stochastic,zheng2023}. It can be observed from Fig.~\ref{fig_correlation_out} the outage probability with separable correlation is greater  than that with non-separable correlation. 
\begin{figure}[t!]
\centering\includegraphics[width=0.5\textwidth]{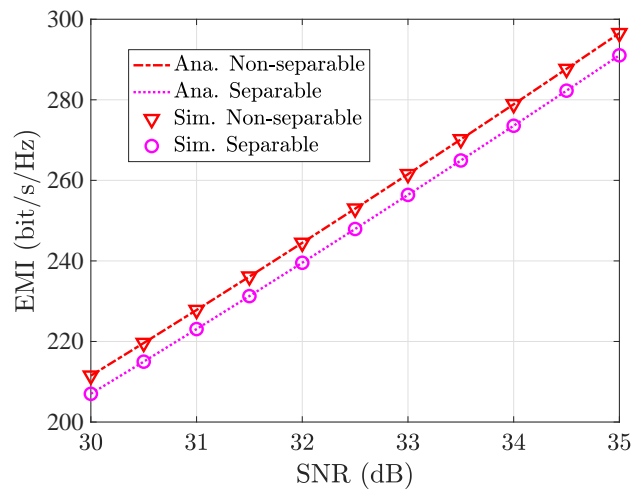}
\caption{Comparison of EMI with separable and non-separable correlations}
\label{fig_correlation_EMI}
\end{figure}

\begin{figure}[t!]
\centering\includegraphics[width=0.5\textwidth]{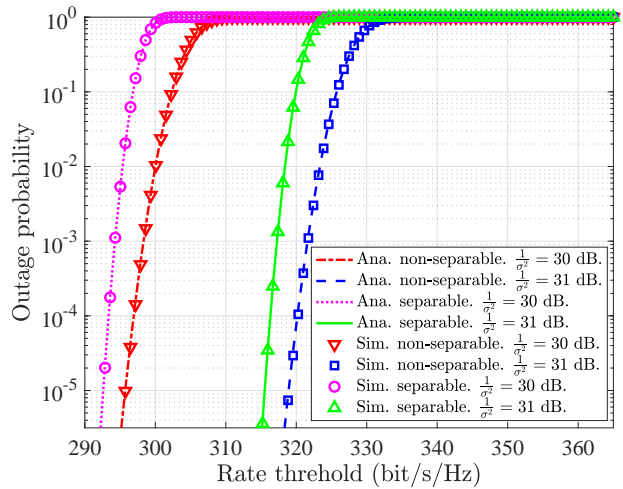}
\caption{Comparison of outage probability with separable and non-separable correlations}
\label{fig_correlation_out}
\end{figure}

\section{Conclusion and Future Works}
\label{sec_con}
In this paper, we investigated the MI over non-centered and non-separable MIMO channels. Specifically, we set up a CLT for the MI and gave the closed-form expressions for the mean and variance. The results can be utilized for evaluating the outage probability of the Rician Weichselberger~\cite{kermoal2002stochastic} and holographic MIMO channels~\cite{pizzo2022fourier}. As far as the authors know, these theoretical results are the first closed-form performance evaluation in the literature for holographic MIMO systems with the non-separable correlation structure. Numerical results validated the accuracy of the evaluation. 

The theoretical results in this paper set up a framework for the analysis of electromagnetically large holographic MIMO channels. The results and approach can be used to perform the finite-blocklength analysis~\cite{zhang2022second} and secrecy analysis~\cite{zhang2022secrecy} for holographic MIMO systems. It is worth mentioning that the analysis in this paper was performed based on the assumption that perfect CSI is available at the receiver and there is no noise among the coupling antennas. It is of practical interest to investigate the impact of the imperfect CSI~\cite{an2023tutorial} and noise in the coupling antennas~\cite{morris2005network } on the holographic MIMO systems, which will be considered in future works.

\newpage

\appendices

\section{Fourier Plane-wave Representation of Electromagnetic Channels}
\label{fourier_details}
The authors of \cite{pizzo2022spatial} showed that when the reactive propagation in the proximity (i.e., at a distance of few wavelengths) of the source and scatterers are excluded, the channel response between the transmit antenna at $\bold{s}$ and the receive antenna at $\bold{r}$, i.e., $h_s(\bold{r},\bold{s})$, of two infinitely large arrays can be modeled as a spatially-stationary electromagnetic random field
\begin{equation}
\label{contin_exp}
\begin{aligned}
&h_{s}(\bold{r},\bold{s})=\frac{1}{4\pi^2} \iiiint_{\mathcal{D}_{k}\times\mathcal{D}_{\kappa}} a_{R}(k_x,k_y,\bold{r})
H_{a}(k_x,k_y,\kappa_x,\kappa_y)a_{S}^{*}(\kappa_x,\kappa_y,\bold{r})\mathrm{d} k_x\mathrm{d} k_y \mathrm{d} \kappa_x \mathrm{d}\kappa_y,
\end{aligned}
\end{equation}
where $a_{S}(\kappa_x,\kappa_y,\bold{s})=e^{\jmath \boldsymbol{\kappa}^{T}\bold{s}}=e^{-\jmath(\kappa_x s_x+\kappa_y s_y +\kappa_z s_z)}$ and $a_{R}(k_x,k_y,\bold{r})=e^{\jmath \bold{k}^{T}\bold{r}}=e^{\jmath(k_x r_x+k_y r_y +k_z r_z)}$
represent the source response at $\bold{s}$ and the receive response at $\bold{r}$, respectively. Here, $\boldsymbol{\kappa}=[\kappa_x,\kappa_y,\kappa_z]^{T}$ with $\kappa_z=\gamma(\kappa_x,\kappa_y)=(\kappa^2-\kappa_x^2-\kappa_y^2)^{\frac{1}{2}}$ and the wave number $\kappa$ is given by $\kappa=\frac{2\pi}{\lambda}$ and $\boldsymbol{\kappa}/\|\boldsymbol{\kappa} \|$ represents the propagation direction of the transmit field. Similarly, $\boldsymbol{k}=[k_x,k_y,k_z]^{T}$ with $k_z=\gamma(k_x,k_y)=(\kappa^2-k_x^2-k_y^2)^{\frac{1}{2}}$ and $\boldsymbol{k}/\|\boldsymbol{k} \|$ represents the propagation direction of the receive field. The integration region $\mathcal{D}_{k}$ is given by the circular area with radius $\kappa$,
\begin{equation}
\label{region1}
\mathcal{D}_{k}=\left\{ (k_x,k_y)\in \mathbb{R}^{2}  | k_x^2+k_y^2\le \kappa^2   \right\}.
\end{equation}  
This is because only the propagating waves are considered and the evanescent waves are neglected~\cite{pizzo2022spatial,pizzo2022fourier}. $H_{a}(k_x,k_y,\kappa_x,\kappa_y)$ represents the angular response from the transmit direction $\boldsymbol{\kappa}/\|\boldsymbol{\kappa}\|$ to the receive direction $\bold{k}/ \|\bold{k}\|$, which is determined by the scattering environment and the array geometry. In (\ref{contin_exp}), $[\BH_s]_{i,j}=h_s(\bold{r}_i,\bold{s}_{j})$ is obtained by the continuous Fourier spectral representation. When both the LoS and non-LoS components are considered, the Fourier coefficient ${H}(l_x,l_y,m_x,m_y)$ can be written as
\begin{equation}
\label{Ha_ele}
\begin{aligned}
H(l_x,l_y,m_x,m_y)&={A}(l_x,l_y,m_x,m_y)
+{Y}(l_x,l_y,m_x,m_y),
\end{aligned}
\end{equation}
where ${A}(l_x,l_y,m_x,m_y)$ and ${Y}(l_x,l_y,m_x,m_y)$ denote the coefficients of the LoS and NLoS component, respectively. Under such circumstances, the channel matrix in spatial domain $\BH_s$ can be represented by~(\ref{hs_fourier}).

\section{Proof of Lemma~\ref{sum_mart}}
\label{proof_sum_mart}
The proof of Lemma~\ref{sum_mart} relies on a system of equations, which can be derived by resolvent computations. The proof idea can be summarized as:

1. We first evaluate each term at the left hand side of~(\ref{sum_pjj}), which consists of $\Theta_{j,i}$ and $\psi_{j,i}$. Specifically, we set up two groups of system of equations to solve $\Theta_{j,i}$ and $\psi_{j,i}$. The first group of equations is obtained by $\E\Ba^{H}_k (\E_{j}\BQ)\BD_{j}\BQ\Ba_{k} $ using two approaches and which are given in Appendix~\ref{sec_first_eq}.

2. The second group of equations is obtained by evaluating $\frac{1}{M}\E \Tr(\E_j\BQ)\BD_j\BQ\BD_i$ with rank-one perturbations, which is given in Appendix~\ref{sec_second_eq}.

3. By combining the two groups of equations in 1 and 2, we can obtain the closed-form approximation for $\Theta_{j,j}$ and $\psi_{j,j}$ based on the Crammer's rule so that we can compute $\sum_{j=1}^{M} (\frac{\rho^2\widetilde{t}_{j,j}^2 \psi_{j,j}}{M}+\frac{2\Theta_{j,j}}{M} )$. The details are given in Appendix~\ref{sec_eva_sum_mart}.

Before we proceed, we first introduce some useful results that will be used in the following derivation. The resolvent related results are given below
\begin{equation}
\label{ide_Q_QI}
\BQ=\BQ_{i}-\frac{\BQ_{i}\Bh_i\Bh_i^{H}\BQ_{i}}{1+\Bh_i^{H}\BQ_{i}\Bh_i}
\end{equation}
\begin{equation}
\label{h_Q_QI}
\bold{h}_i^{H}\BQ=\rho\widetilde{q}_{i,i}\BQ_{i}=\frac{\Bh_i^{H}\BQ_i}{1+\Bh_i^{H}\BQ_i\Bh_i},
\end{equation}
where $\widetilde{q}_{i,i}$ is the $i$-th diagonal entry of $\widetilde{\BQ}$ and can be computed by~(\ref{qii_exp}). Define
\begin{equation}
\BT_i=\left(  \boldsymbol{\psi}(z) -z\bold{A}_i  \widetilde{\boldsymbol{\psi}}_i(z)   \bold{A}_i^{H}     \right)^{-1},
\end{equation}
where $\widetilde{\boldsymbol{\psi}}_i(z)$ is obtained by removing the $i$-th row and the $i$-th column from $\widetilde{\boldsymbol{\psi}}(z)$ and $\BA_i$ is obtained by removing the $i$-th column from $\BA$. There holds true 
\begin{equation}
\label{rho_aTb}
\rho\widetilde{t}_{i,i}\Ba_i^{H}\bold{T}_i \bold{b}=\frac{\Ba_i^{H}\BT\bold{b}}{1+\delta_i},
\end{equation}
\begin{equation}
\label{t_jj_exp}
\widetilde{t}_{i,i}=\rho^{-1}\left( 1 +  \delta_i + \Ba_i^{H}\BT_i\Ba_i \right)^{-1},
\end{equation}
\begin{equation}
\label{t_jj_exp2}
\rho\widetilde{t}_{i,i} \Ba_i^{H}\BT_i\Ba_i =1-\rho(1+\delta_i)\widetilde{t}_{i,i}.
\end{equation}
(\ref{rho_aTb}) can be obtained by similar steps in~\cite[Appendix A]{hachem2012clt}.

\subsection{The First Group of Equations}
\label{sec_first_eq}
In this part, we will utilize two approaches to evaluate $\E\Ba^{H}_k (\E_{j}\BQ)\BD_{j}\BQ\Ba_{k} $ to set up the first group of equations by resolvent computations.
\subsubsection{First evaluation of $\E\Ba^{H}_k (\E_{j}\BQ)\BD_{j}\BQ\Ba_{k} $}
By applying~(\ref{ide_Q_QI}) for both $\BQ$ in $\E\Ba^{H}_k (\E_{j}\BQ)\BD_{j}\BQ\Ba_{k} $, $\E\Ba^{H}_k (\E_{j}\BQ)\BD_{j}\BQ\Ba_{k} $ can be rewritten as
\begin{equation}
\begin{aligned}
\E\Ba^{H}_k (\E_{j}\BQ)\BD_{j}\BQ\Ba_{k}&=\E\Ba^{H}_k (\E_{j}\BQ_k)\BD_{j}\BQ_k\Ba_{k}
-\rho\E\Ba^{H}_k (\E_{j}\widetilde{q}_{k,k} \BQ_k \Bh_k\Bh_k^{H} \BQ_k)\BD_{j}\BQ_j\Ba_{k}
\\
&-\rho\E\widetilde{q}_{k,k}\Ba^{H}_k (\E_{j} \BQ_k )\BD_{j}\BQ_k\Bh_k\Bh_k^{H} \BQ_k\Ba_{k}
+\rho^2\E\widetilde{q}_{k,k}^2\Ba^{H}_k (\E_{j} \BQ_k\Bh_k\Bh_k^{H} \BQ_k )\BD_{j}\BQ_k\Bh_k\Bh_k^{H} \BQ_k\Ba_{k}
\\
&=\theta_{j,k}+U_1+U_2+U_3,
\end{aligned}
\end{equation}
where $\theta_{j,k}=\E\Ba_{k}^{H}(\E_j \BQ_k)\BD_j\BQ_k\Ba_k$. We can safely replace $\widetilde{q}_{k,k}$ by $\widetilde{t}_{k,k}$ in $U_1$ based on the following fact
\begin{equation}
\label{aQhhQDQa}
\begin{aligned}
&|\rho\E\Ba^{H}_k (\E_{j}(\widetilde{q}_{k,k}-\widetilde{t}_{k,k}) \BQ_k \Bh_k\Bh_k^{H} \BQ_k)\BD_{j}\BQ_k\Ba_{k}|
\overset{(a)}{\le} \E \rho K \E^{\frac{1}{2}}_{j} (\widetilde{q}_{k,k}-\widetilde{t}_{k,k})^2 \E^{\frac{1}{2}}_j (\Bh_k^{H}\Bh_k)^2
\\  
&\le \rho K' \E^{\frac{1}{2}} (\widetilde{q}_{k,k}-\widetilde{t}_{k,k})^2 \E^{\frac{1}{2}} (\Bh_k^{H}\Bh_k)^2\overset{(b)}{=}o(1),
\end{aligned}
\end{equation}
where step $(a)$ in~(\ref{aQhhQDQa}) follows from the Cauchy-Schwarz inequality and 
\begin{equation}
\begin{aligned}
& \Bh_k^{H} \BQ_k\BD_{j}\BQ_k\Ba_{k}\Ba_{k}^{H} \BQ_k\BD_{j}\BQ_k\Bh_k
 \le \Bh_k^{H}\Bh_k \| \BQ_k\BD_{j}\BQ_k\Ba_{k}\Ba_{k}^{H} \BQ_k\BD_{j}\BQ_k \|
  \le \rho^{-4}\sigma^4_{max}a_{max}^2\Bh_k^{H}\Bh_k.
\end{aligned}
 \end{equation}
 Step $(b)$ in~(\ref{aQhhQDQa}) follows from the approximation of $\widetilde{t}_{k,k}$ in~(\ref{qjj_approx}) of Lemma~\ref{appro_diag_ele}. Therefore, we have
\begin{equation}
\label{U_1_approx}
\begin{aligned}
U_1&=-\rho\widetilde{t}_{k,k}\E\Ba^{H}_k (\E_{j} \BQ_k \Bh_k\Bh_k^{H} \BQ_k)\BD_{j}\BQ_j\Ba_{k}
\overset{(a)}{=}-\rho\widetilde{t}_{k,k}\Ba^{H}_k\BT_k \Ba_k\Ba_k^{H}(\E \BQ_k)\BD_{j}\BQ_k\Ba_{k}+o(1)
\\
&
=-\rho\widetilde{t}_{k,k}\Ba^{H}_k\BT_k \Ba_k\theta_{k,j}+o(1),
\end{aligned}
\end{equation}
where step $(a)$ in~(\ref{U_1_approx}) follows from $\E|  \Ba_k\BQ_k\Ba_k- \Ba_k\BT_k\Ba_k |^2=o(1)$, which is obtained by~(\ref{uqvcon}) in Lemma~\ref{appro_diag_ele}. Similarly, we can replace $\widetilde{q}_{k,k}$ by $\widetilde{t}_{k,k}$ in $U_2$ and $U_3$ to obtain
\begin{equation}
\begin{aligned}
U_2&=-\rho\E\widetilde{q}_{k,k}\Ba^{H}_k (\E_{j} \BQ_k )\BD_{j}\BQ_k\Bh_k\Bh_k^{H} \BQ_k\Ba_{k}
\\
&
=-\rho\widetilde{t}_{k,k}\Ba^{H}_k\BT_k \Ba_k\E\Ba_k^{H}(\E \BQ_k)\BD_{j}\BQ_k\Ba_{k}+o(1)
\\
&
=-\rho\widetilde{t}_{k,k}\Ba^{H}_k\BT_k \Ba_k\theta_{j,k}+o(1)
\end{aligned}
\end{equation}
and
\begin{equation}
\begin{aligned}
U_3&=\rho^2\widetilde{t}_{k,k}^2(\Ba^{H}_k \BT_k\Ba_k)^2 (\frac{\E \Tr  (\E_j \BQ_k) \BD_{j}\BQ_k\BD_{k}}{M}
+\E \Ba_k^{H}(\E_j \BQ_k) \BD_{j}\BQ_k\Ba_k ),~k \le j,
\end{aligned}
\end{equation}
Therefore, by $1-\rho\widetilde{t}_{k,k}\Ba_k\BT_k\Ba_k=\rho\widetilde{t}_{k,k}(1+\delta_k)$, we have 
\begin{equation}
\label{eva_aqdq_1}
\begin{aligned}
&\E\Ba^{H}_k (\E_{j}\BQ)\BD_{j}\BQ\Ba_{k}=(1+\delta_k)^2\rho^2\widetilde{t}_{k,k}^2 \theta_{j,k} 
+\frac{(\Ba_k^{H}\BT\Ba_k)^2}{(1+\delta_k)^2} \psi_{j,k}+o(1),~k\le j.
\end{aligned}
\end{equation}
By far, we have obtained the first evaluation for $\E\Ba^{H}_k (\E_{j}\BQ)\BD_{j}\BQ\Ba_{k}$.

\subsubsection{Second evaluation of $\E\Ba^{H}_k (\E_{j}\BQ)\BD_{j}\BQ\Ba_{k} $}
In this part, we will $\E\Ba^{H}_k (\E_{j}\BQ)\BD_{j}\BQ\Ba_{k}$ by another approach. First, by the identity $\BA-\BB=\BB(\BB^{-1}-\BA^{-1}) \BA$, we can obtain
\begin{equation}
\label{aqdqak}
\begin{aligned}
\E\Ba^{H}_k (\E_{j}\BQ)\BD_{j}\BQ\Ba_{k} 
&=\E\Ba^{H}_k \BT\BD_{j}\BQ\Ba_{k}
+\rho\E\Ba^{H}_k  \BT\bold{\Upsilon}(\E_{j} \BQ) \BD_{j}\BQ\Ba_{k}
+\E\Ba^{H}_k\BT \BA \widetilde{\bold{\Psi}}\BA^{H} (\E_{j}\BQ) \BD_{j}\BQ\Ba_{k}
\\
&
-\E\Ba^{H}_k \BT (\E_{j}\BH\BH^{H} \BQ)\BD_{j}\BQ\Ba_{k} 
\\
&=W_1+W_2+W_3+W_4,
\end{aligned}
\end{equation}
where $\bold{\Upsilon}=\diag(\widetilde{\delta}_{1},\widetilde{\delta}_{2},...,\widetilde{\delta}_{N})$. By applying~(\ref{ide_Q_QI}) to $(\E_j \BQ)$, $W_3$ can be rewritten as 
\begin{equation}
\begin{aligned}
W_3&=\sum_{i=1}^{M} \frac{ \E\Ba^{H}_k \BT\Ba_{i}\Ba_{i}^{H} (\E_{j}\BQ_{i})\BD_{j}\BQ\Ba_{k}}{1+\delta_{i}} 
- \frac{\rho  \Ba^{H}_k \BT\Ba_{i}\E \Ba_{i}^{H}   (\E_{j}\widetilde{q}_{i,i}\BQ_{i}\Bh_{i}\Bh_{i}^{H} \BQ_{i})\BD_{j}\BQ\Ba_{k} }{ 1+\delta_{i} }
\\
&=\sum_{i=1}^{M} \frac{\Ba^{H}_k \BT\Ba_{i}\E \Ba_{i}^{H}(\E_{j}\BQ_{i})\BD_{j}\BQ\Ba_{k}}{1+\delta_{i}} 
- \frac{\rho\widetilde{t}_{i,i}    \Ba^{H}_k \BT  \Ba_{i} \E\Ba_{i}^{H} (\E_{j}\BQ_{i}\Bh_{i}\Bh_{i}^{H} \BQ_{i})\BD_{j}\BQ\Ba_{k}}{1+\delta_{i}}+\varepsilon_{3,1}
\end{aligned}
\end{equation}
where
\begin{equation}
\label{eps_31}
\begin{aligned}
&\varepsilon_{3,1}=\E\sum_{i=1}^{M}\frac{ \rho  \Ba_{i}   (\E_{j}(\widetilde{t}_{i,i}-\widetilde{q}_{i,i})\BQ_{i}\Bh_{i}\Bh_{i}^{H} \BQ_{i})\BD_{j}\BQ\Ba_{k}}{1+\delta_{i}}
 \Ba^{H}_k \BT\Ba_{i} 
= \rho  \Ba^{H}_k \BT\BA\bold{\Upsilon}\bold{\Lambda}_{t,q}\BH^{H}\BQ\BD_j\BQ\Ba_{k},
\end{aligned}
\end{equation}
and $\bold{\Lambda}_{t,q}=\diag((\widetilde{t}_{i,i}-\widetilde{q}_{i,i})(1+\Bh_i^{H}\BQ_i\Bh_i)\BQ_i\Bh_i )$, $i=1,2,...,M$. By~(\ref{qjj_approx}) in Lemma~\ref{appro_diag_ele} and the definition of the matrix norm $ \|\BA\bold{a}\|_2 \le \| \BA\| \| \bold{a}\|_{2}$, we can obtain 
\begin{equation}
\begin{aligned}
\label{eps_31}
|\varepsilon_{3,1}|&\le \rho\E (\E_{j}  \|   \Ba^{H}_k \BT\BA\bold{\Upsilon}\bold{\Lambda}_{t,q} \BH^{H}\BQ\| )\| \BD_j\BQ\Ba_{k} \|_{2}
\\
&
\overset{(a)}{\le} L \E  \|   \Ba^{H}_k \BT\BA\bold{\Upsilon}\bold{\Lambda}_{t,q} \BH^{H}\BQ\| 
\\
&
\overset{(b)}{\le}  L  \left(\sum_{l=1}^{M}|\Ba^{H}_k \BT\Ba_{l}|^{2} \E(\widetilde{t}_{i,i}-\widetilde{q}_{i,i})^2   \right)^{\frac{1}{2}}
\\
&\overset{(c)}{\le} L  (\Ba^{H}_k \BT\BA\BA^{H}\BT\Ba_k  )^{\frac{1}{2}} \times o(1),
\end{aligned}
\end{equation}
where step $(a)$ follows from the evaluation $\| \BD_j\BQ\Ba_{k} \|_{2}\le \| \BD_j\BQ \|  \|\Ba_{k}\|_{2} \le \frac{\sigma_{max}^2 a_{max}}{\rho}:=L$ and $\|\BH^{H}\BQ \| <1$. Steps $(b)$ and $(c)$ follow from the Cauchy-Schwarz inequality and~(\ref{qjj_approx}) in Lemma~\ref{appro_diag_ele}. $X_3$ can be further rewritten as
\begin{equation}
\begin{aligned}
W_3&=\sum_{i=1}^{M} \frac{\E  \Ba^{H}_k \BT\Ba_{i} \Ba_{i}(\E_{j}\BQ_{i})\BD_{j}\BQ\Ba_{k} }{1+\delta_{i}}
- \frac{\E \rho \widetilde{t}_{i,i}   \Ba^{H}_k \BT  \Ba_{i} \Ba_{i} (\E_{j}\BQ_{i}\Ba_{i}\Bh_{i}^{H} \BQ_{i})\BD_{j}\BQ\Ba_{k}}{1+\delta_{i}}
+\varepsilon_{3,2}+o(1)
\\
&
=\sum_{i=1}^{M} \frac{ \Ba^{H}_k \BT\Ba_{i} \E\Ba_{i}^{H}(\E_{j}\BQ_{i})\BD_{j}\BQ\Ba_{k} }{1+\delta_{i}}
- \frac{\E\rho   \widetilde{t}_{i,i}\Ba^{H}_k \BT \Ba_{i} \Ba_{i} \BT_{i}\Ba_{i}(\E_{j}\Bh_{i}^{H} \BQ_{i})\BD_{j}\BQ\Ba_{k} }{1+\delta_{i}}
+\varepsilon_{3,2}
+\varepsilon_{3,3}+o(1)
\\
&
=W_{3,1}+W_{3,2}+\varepsilon_{3,2}+\varepsilon_{3,3}+o(1),
\end{aligned}
\end{equation}
where 
\begin{equation}
\begin{aligned}
\varepsilon_{3,3} &=-\sum_{i=1}^{M}\rho(1+\delta_{i})^{-1} \widetilde{t}_{i,i}  
\E\Ba^{H}_k \BT  \Ba_{i} \Ba_{i}^{H} (\E_{j}\BQ_{i}\By_{i}\Bh_{i}^{H} \BQ_{i})\BD_{j}\BQ\Ba_{k}
\\
\varepsilon_{3,2}&=-\sum_{i=1}^{M}\rho(1+\delta_{i})^{-1}   \widetilde{t}_{i,i}\Ba^{H}_k \BT \Ba_{i} 
\Ba_{i}^{H}(\E_{j}( \BT_{i}-\BQ_i)\Ba_{i}\Bh_{i}^{H} \BQ_{i})\BD_{j}\BQ\Ba_{k},
\end{aligned}
\end{equation}
which can be shown to be $o(1)$ by a similar approach for $\varepsilon_{3,1}$ shown in~(\ref{eps_31}). Now we turn to evaluate $W_{3,2}$ as
\begin{equation}
\begin{aligned}
W_{3,2} & \overset{(a)}{=}-\sum_{i=1}^{M}  \frac{\rho   \widetilde{t}_{i,i}  \Ba^{H}_k \BT \Ba_{i} \Ba_{i}^{H} \BT_{i}\Ba_{i}\E\Ba_{i}^{H}(\E_{j} \BQ_{i})\BD_{j}\BQ\Ba_{k} }{1+\delta_{i}}
-\sum_{i=1}^{j} \frac{\rho  \widetilde{t}_{i,i}  \Ba^{H}_k \BT \Ba_{i} \Ba_{i}^{H} \BT_{i}\Ba_{i}\E\By_{i}^{H}(\E_{j} \BQ_{i})\BD_{j}\BQ\Ba_{k} }{1+\delta_{i}}
\\
&=W_{3,2,1}+W_{3,2,2}
\end{aligned}
\end{equation}
where step $(a)$ follows from the fact that for $i=j+1,...,M$,
\begin{equation}
\widetilde{t}_{i,i}  \Ba^{H}_k \BT \Ba_{i} \Ba_{i} \BT_{i}\Ba_{i}(\E_{j}\By_{i}^{H} \BQ_{i})\BD_{j}\BQ\Ba_{k}=0.
\end{equation}
By applying~(\ref{ide_Q_QI}) to $\BQ$, we have
\begin{equation}
\label{W_322}
\begin{aligned}
W_{3,2,2} &=-\sum_{i=1}^{j} \E[\frac{\rho   \widetilde{t}_{i,i}  \Ba^{H}_k \BT \Ba_{i} \Ba_{i} \BT_{i}\Ba_{i} }{1+\delta_{i}} \By_{i}^{H}(\E_{j} \BQ_{i})\BD_{j}(\BQ_i- \rho\widetilde{q}_{i,i}\BQ_i\Bh_i \Bh_i^{H}\BQ_i)\Ba_{k}]
\\
&=\sum_{i=1}^{j}\frac{  \rho^2 \widetilde{t}_{i,i}^2  \Ba^{H}_k \BT \Ba_{i} \Ba_{i}^{H} \BT_{i}\Ba_{i}\Ba_{i}^{H} \BT_{i}\Ba_{k}   }{ (1+\delta_{i})} \psi_{j,i}+o(1)
\\
&\overset{(a)}{=}\sum_{i=1}^{j} \frac{\rho^2   \widetilde{t}_{i,i}^2  \Ba^{H}_k \BT_i \Ba_{i} \Ba_{i}^{H} \BT_{i}\Ba_{i}\Ba_{i}^{H} \BT\Ba_{k}}{(1+\delta_{i})^2}\psi_{j,i} +o(1),
\end{aligned}
\end{equation}
where step $(a)$ in~(\ref{W_322}) follows from the adaptation of~\cite[Lemma 5.2]{hachem2012clt}. By~(\ref{h_Q_QI}) and $(\E_{j}\Bh_{i}\Bh_{i}^{H} \BQ_{i})=(\E_{j}\Ba_{i}\Ba_{i}^{H} \BQ_{i} + \frac{1}{M}\BD_j\BQ_i  ),~j< i \le M$, we have the evaluation for $W_4$ as
\begin{equation}
\label{w4_eva}
\begin{aligned}
W_4&=-\sum_{i=1}^{M}\E \rho  \Ba^{H}_k \BT (\E_{j}\widetilde{q}_{i,i}\Bh_{i}\Bh_{i}^{H} \BQ_{i})\BD_{j}\BQ\Ba_{k} 
=-\sum_{i=1}^{M}\E \rho \widetilde{t}_{i,i} \Ba^{H}_k \BT (\E_{j}\Bh_{i}\Bh_{i}^{H} \BQ_{i})\BD_{j}\BQ\Ba_{k}+o(1)
\\
&=-\sum_{i=1}^{M}\E \rho \widetilde{t}_{i,i} \Ba^{H}_k \BT \Ba_{i}\Ba_{i}^{H} (\E_{j}\BQ_{i})\BD_{j}\BQ\Ba_{k}
-(\sum_{i=1}^{j} \rho \widetilde{t}_{i,i}\E  \Ba^{H}_k  \BT\By_i\By_i^{H} (\E_{j}\BQ_i)\BD_{j}\BQ\Ba_{k}
\\
&
+\sum_{i=j+1}^{M}\frac{ \rho \widetilde{t}_{i,i}\E  \Ba^{H}_k  \BT \BD_{i} (\E_{j}\BQ_i)\BD_{j}\BQ\Ba_{k}}{M})
-\sum_{i=1}^{j} \rho \widetilde{t}_{i,i}\E  \Ba^{H}_k  \BT\By_i\Ba_i^{H} (\E_{j}\BQ_i)\BD_{j}\BQ\Ba_{k}
\\
&
-\sum_{i=1}^{j} \rho \widetilde{t}_{i,i}\E  \Ba^{H}_k  \BT\Ba_i\By_i^{H} (\E_{j}\BQ_i)\BD_{j}\BQ\Ba_{k}
+o(1)
\\
&=W_{4,1}+W_{4,2}+W_{4,3}+W_{4,4}+o(1).
\end{aligned}
\end{equation}
By~(\ref{h_Q_QI}), we have
\begin{equation}
\label{W42_eva}
\begin{aligned}
W_{4,2} &=-\sum_{i=1}^{j} \rho \widetilde{t}_{i,i}\E  \Ba^{H}_k  \BT\By_i\By_i^{H} (\E_{j}\BQ_i)\BD_{j}\BQ_i\Ba_{k}
+\sum_{i=1}^{j} \rho^2 \widetilde{t}_{i,i}^2\E  \Ba^{H}_k  \BT\By_i\By_i^{H} (\E_{j}\BQ_i)\BD_{j}\BQ_i \Bh_i\Bh_i^{H} \BQ_i\Ba_{k} 
\\
&
-\sum_{i=j+1}^{M}\frac{ \rho \widetilde{t}_{i,i}\E  \Ba^{H}_k  \BT \BD_{i} (\E_{j}\BQ_i)\BD_{j}\BQ\Ba_{k}}{M}
+o(1)
\\
&
\overset{(a)}{=}\sum_{i=1}^{j} \frac{\rho^2 \widetilde{t}_{i,i}^2\E  \Ba^{H}_k  \BT\BD_i  \BT\Ba_{k}  \Tr\BD_i(\E_{j}\BQ_i)\BD_{j}\BQ_i  }{M^2}
-\sum_{i=1}^{M}\frac{ \rho \E  \Ba^{H}_k  \BT{\bold{\Upsilon}} (\E_{j}\BQ_i)\BD_{j}\BQ\Ba_{k}}{M}
+\varepsilon_{4,2}+o(1)
\\
&
=\underbrace{\sum_{i=1}^{j}  \frac{\rho^2 \widetilde{t}_{i,i}^2  \Ba^{H}_k  \BT\BD_i  \BT\Ba_{k} }{M}\psi_{j,i}}_{W_{4,2,1}}+W_{4,2,2}
+\varepsilon_{4,2}+o(1)
,
\end{aligned}
\end{equation}
where step $(a)$ in~(\ref{W42_eva}) is obtained by $\frac{1}{M}\sum_{i}^{M}\widetilde{t}_{i,i}[\BD_i]_{k,k}=\frac{1}{M}\Tr\widetilde{\BD}_k\widetilde{\BT} $. We can observed $W_{4,2,2}=-\rho\Ba^{H}_k  \BT\bold{\Upsilon}(\E_{j} \BQ) \BD_{j}\BQ\Ba_{k}=-W_{2}$. Here, $\varepsilon_{4,2}$ can be bounded by the Cauchy-Schwarz inequality as
\begin{equation}
\label{ep_42}
\begin{aligned}
|\varepsilon_{4,2}| &=|\sum_{i=1}^{j}\frac{ \rho \widetilde{t}_{i,i}}{M}\E  \Ba^{H}_k  \BT\BD_i (\E_{j}\BQ_i)\BD_{j}(\BQ-\BQ_i)\Ba_{k}|
\le \frac{K}{M}\sum_{i=1}^{j} \E^{\frac{1}{2}} (1+\Bh_i\BQ_i\Bh_i)^2  \E^{\frac{1}{2}} | \Ba_i^{H}\BQ\Bh_i |^2
\\
&
\overset{(a)}{\le} \frac{K' \sqrt{j}}{M}    \E^{\frac{1}{2}} \Ba_i^{H}\BQ\BH_{[1:j]}\BH_{[1:j]}^{H}\BQ\Ba_i=\BO(M^{-\frac{1}{2}}),
\end{aligned}
\end{equation}
with $\BH_{[1:j]}$ being the matrix consisting of the first $j$ columns of $\BH$. By similar analysis of~(\ref{ep_42}), we can evaluate $W_{4,3}$ as
\begin{equation}
\begin{aligned}
W_{4,3}&=\sum_{i=1}^{j} \frac{\rho^2 \widetilde{t}_{i,i}^2\Ba^{H}_k  \BT \BD_i \BT\Ba_{k} }{M}
 \E\Ba_i^{H} (\E_{j}\BQ_i)\BD_{j}\BQ_i\Ba_i  +o(1)
\\
&
=\sum_{i=1}^{j}  \frac{\rho^2 \widetilde{t}_{i,i}^2\Ba^{H}_k  \BT \BD_i \BT\Ba_{k} }{M}\theta_{j,i} +o(1).
\end{aligned}
\end{equation}
By~(\ref{ide_Q_QI}), $W_{4,4}$ can be evaluated as
\begin{equation}
\begin{aligned}
W_{4,4}&=\sum_{i=1}^{j} \frac{\rho \widetilde{t}_{i,i} \Ba^{H}_k  \BT\Ba_i  \Ba^{H}_i  \BT_i\Ba_k  }{M}\E  \Tr(\E_{j}\BQ_i)\BD_{j}\BQ_i\BD_i
+o(1)
\\
&
=\sum_{i=1}^{j} \frac{\rho \widetilde{t}_{i,i} \Ba^{H}_k  \BT\Ba_i  \Ba^{H}_i  \BT_i\Ba_k  }{(1+\delta_i)}\psi_{j,i} +o(1).
\end{aligned}
\end{equation}
By~(\ref{t_jj_exp2}), we can obtain
\begin{equation}
\begin{aligned}
W_{3,2,1}+W_{4,1}
&=-\sum_{i=1}^{M}\E \rho \widetilde{t}_{i,i}(\frac{\Ba_{i}^{H} \BT_{i}\Ba_{i}}{1+\delta_{i}}+1) \Ba^{H}_k \BT \Ba_{i}\Ba_{i}^{H} (\E_{j}\BQ_{i})\BD_{j}\BQ\Ba_{k}
\\
&
=\sum_{i=1}^{M}\E \rho \widetilde{t}_{i,i}(1+\delta_{i})^{-1} \Ba^{H}_k \BT \Ba_{i}\Ba_{i}^{H} (\E_{j}\BQ_{i})\BD_{j}\BQ\Ba_{k}
\\
&
=-W_{3,1}.
\end{aligned}
\end{equation}
According to the evaluation from~(\ref{aqdqak}) to~(\ref{w4_eva}), $\E\Ba^{H}_k (\E_{j}\BQ)\BD_{j}\BQ\Ba_{k} $ can be represented by
\begin{equation}
\label{eva_aqdq_2}
\begin{aligned}
\E\Ba^{H}_k (\E_{j}\BQ)\BD_{j}\BQ\Ba_{k} & =W_{4,3}+(W_{4,2,1}
+W_{3,2,2}+W_{1}
+W_{4,4})+o(1)
\\
& =\sum_{i=1}^{j}\frac{\rho^2 \widetilde{t}_{i,i}^2  \Ba^{H}_k \BT\BD_{i}\BT\Ba_{k}}{M}\theta_{j,i}
+\sum_{i=1}^{j} [\frac{\rho^2 \widetilde{t}_{i,i}^2 \Ba^{H}_k \BT\BD_{i}\BT\Ba_{k}}{M}   +   \frac{ \Ba^{H}_k \BT_i \Ba_{i} \Ba_{i}^{H} \BT_{i}\Ba_{k}}{(1+\delta_{i})^{2} }  ] \psi_{j,i}
+\Ba^{H}_k \BT\BD_{j}\BT\Ba_{k}
+o(1).
\end{aligned}
\end{equation}
By the evaluation for $\E\Ba^{H}_k (\E_{j}\BQ)\BD_{j}\BQ\Ba_{k}$ in~(\ref{eva_aqdq_1}) and~(\ref{eva_aqdq_2}), we can set up the first group equations
\begin{equation}
\begin{aligned}
\label{eq_first}
 &  \sum_{i=1}^{j} (-\rho^2 \widetilde{t}_{i,i}^2\Pi_{i,k}  -\mathbbm{1}_{k\neq i}\Xi_{i,k}  )   \psi_{j,i}
    +  \sum_{i=1}^{j}( \mathbbm{1}_{k=i} -\Pi_{i,k})\rho^2 \widetilde{t}_{i,i}^2\theta_{j,i} 
   =M\Pi_{j,k}.
\end{aligned}
\end{equation}
Next, we will set up the second group of equations by evaluating $\frac{1}{M}\E \Tr(\E_j \BQ) \BD_{j}\BQ\BD_i$.

\subsection{The Second Group of Equations}
\label{sec_second_eq}
To set up the second group of equations, we first evaluate $\frac{1}{M}\E \Tr(\E_j \BQ) \BD_{j}\BQ\BD_i$ by resolvent computations, which can be approximated by a linear combination of $\Theta_{j,k}$ and $\psi_{j,k}$. Then we approximate $\frac{1}{M}\E \Tr(\E_j \BQ) \BD_{j}\BQ\BD_i$ by $\psi_{j,i}$ to obtain the system of equations.

By applying the identity $\BA-\BB=\BB(\BB^{-1}-\BA^{-1}) \BA$, we can obtain
\begin{equation}
\begin{aligned}
\label{psi_first}
\frac{1}{M}\E \Tr(\E_j \BQ) \BD_{j}\BQ\BD_i &=
\frac{\E\Tr\BT\BD_j\BQ\BD_i}{M}+ \frac{\rho \E\Tr \BT\bold{\Upsilon}(\E_{j} \BQ) \BD_{j}\BQ\BD_i}{M}
\\
&+\frac{\E\Tr\BT \BA \widetilde{\bold{\Psi}}\BA^{H} (\E_{j}\BQ) \BD_{j}\BQ\BD_i}{M}
-\frac{\E\Tr \BT (\E_{j}\BH \BH^{H}\BQ) \BD_{j}\BQ\BD_i}{M}
\\
&=X_{1}+X_{2}+X_{3}+X_{4}.
\end{aligned}
\end{equation}
First we can obtain $X_{1}=M\Gamma_{i,j}+o(1)$ from~(\ref{utraceq}) in Lemma~\ref{appro_diag_ele}. Then we will evaluate $X_{3}$. By~(\ref{ide_Q_QI}) and~(\ref{qjj_approx}) in Lemma~\ref{appro_diag_ele}, we can obtain
\begin{equation}
\begin{aligned}
X_{3} & = \sum_{i=1}^{M}\frac{\E\Ba_k^{H}(\E_{j}\BQ) \BD_{j}\BQ\BD_i\BT\Ba_k}{M(1+\delta_k)}
\\
&
=\E \sum_{i=1}^{M}\frac{\Ba_k^{H}(\E_{j}\BQ_i) \BD_{j}\BQ\BD_i\BT\Ba_k}{M(1+\delta_k)}
- \sum_{i=1}^{M}\frac{\rho\widetilde{t}_{k,k}\Ba_k^{H}(\E_{j} \BQ_k\Bh_k\Bh_k^{H} \BQ_k) \BD_{j}\BQ\BD_i\BT\Ba_k}{M(1+\delta_k)} +o(1)
\\
&=
X_{3,0}
+X_{3,1}+X_{3,2}+X_{3,3}+X_{3,4}+o(1),
\end{aligned}
\end{equation}
where $X_{3,1}$, $X_{3,2}$, $X_{3,3}$, and $X_{3,4}$ are obtained by the expansion of $\Bh_i\Bh_i^{H}=\Ba_i\Ba_i^{H}+\By_i\Ba_i^{H}+\Ba_i\By_i^{H}+\By_i\By_i^{H}$. By applying~(\ref{ide_Q_QI}) to $\BQ$, $X_{3,1}$ can be evaluated by
\begin{equation}
\begin{aligned}
X_{3,1} &= \sum_{i=1}^{M}\frac{-\E\rho\widetilde{t}_{k,k}\Ba_k^{H}(\E_{j} \BQ_k\Ba_k\Ba_k^{H} \BQ_k) \BD_{j}\BQ\BD_i\BT\Ba_k}{M(1+\delta_k)}
\\
&=-  \sum_{i=1}^{M}\!\frac{\rho\widetilde{t}_{k,k}\Ba_k^{H} \BT_k\Ba_i \E \Ba_k^{H} (\E_{j}\BQ_k) \BD_{j}\BQ\BD_i\BT\Ba_k}{M(1+\delta_k)}\!+\! o(1).
\end{aligned}
\end{equation}
By the Cauchy-Schwarz inequality, $\| \BQ\| \le \rho^{-1}$, and $\|\bold{a}_{k} \|_2<a_{max}$, we can show $X_{3,2}$ is $o(1)$ with
\begin{equation}
\label{X_32_eva}
\begin{aligned}
|X_{3,2}| & =  | \sum_{i=1}^{M}\frac{\rho\widetilde{t}_{k,k}\E \Ba_k^{H}(\E_{j} \BQ_k\By_k\Ba_k^{H} \BQ_k) \BD_{j}\BQ\BD_i\BT\Ba_k|}{M(1+\delta_k)}
\\
&\le  |\E  \sum_{i=1}^{M}\frac{K \widetilde{t}_{k,k}\E_{j}^{\frac{1}{2}}( \By_k^{H}\BQ_k\Ba_k\Ba_k^{H}\BQ_k\By_k) | }{M(1+\delta_k)}
=\BO(M^{-1}),
\end{aligned}
\end{equation}
where $K$ is a constant independent of $M$ and $N$. By~(\ref{ide_Q_QI}) and the definition of $\E_j$, $X_{3,3}$ can be evaluated by
\begin{equation}
\label{eva_X33}
\begin{aligned}
X_{3,3} &=-\E  \sum_{k=1}^{M}\frac{\rho\widetilde{t}_{k,k}\Ba_k^{H}(\E_{j} \BQ_k\Ba_k\By_k^{H} \BQ_k) \BD_{j}\BQ\BD_i\BT\Ba_k}{M(1+\delta_k)}
\\
&
=-\E \sum_{k=1}^{j}\frac{\rho\widetilde{t}_{k,k}\Ba_k^{H} \BT_k\Ba_k(\E_{j}\By_k^{H} \BQ_k) \BD_{j}\BQ_i\BD_i\BT\Ba_k}{M(1+\delta_k)}+
  \sum_{k=1}^{j}\frac{\E\rho^2\widetilde{t}_{k,k}^2\Ba_k^{H} \BT_k\Ba_k(\E_{j}\By_k^{H} \BQ_k) \BD_{j}\BQ_k \Bh_k  \Bh_k^{H}\BQ_k \BD_i\BT\Ba_k}{M(1+\delta_k)}
 + o(1) 
 \\
 &
 =  \sum_{k=1}^{j}\frac{\rho^2\widetilde{t}_{k,k}^2\Ba_k^{H} \BT_k\Ba_k \Ba_k^{H}\BT_k \BD_i\BT\Ba_k \psi_{j,k}}{M(1+\delta_k)}
 + o(1)
 \\
 &
 \overset{(a)}{=}\sum_{i=1}^{j}\frac{\Ba_k^{H} \BT\Ba_k \Ba_k^{H}\BT \BD_i\BT\Ba_k \psi_{j,k}}{M(1+\delta_k)^{3}}
 + o(1),
 \end{aligned}
\end{equation}
where step $(a)$ in~(\ref{eva_X33}) follows by~(\ref{rho_aTb}). By similar steps in~(\ref{X_32_eva}), we can also show $X_{3,4}=o(1)$ as follows
 \begin{equation}
\begin{aligned}
| X_{3,4}|&=|\E \sum_{k=1}^{M}\frac{\widetilde{t}_{k,k}\Ba_i^{H}(\E_{j} \BQ_k\By_k\By_k^{H} \BQ_k) \BD_{j}\BQ\BD_i\BT\Ba_k}{M(1+\delta_k)}|
\\
&
 \le \E  K M^{-1} \E_{j}^{\frac{1}{2}}  (|\Ba_k^{H}\BQ_k\By_k|^{2}  \By_k^{H}\By_k )
 \\
 &
 \le K' M^{-1} \sum_{i=1}^{M} ( \E (\By_k^{H}\By_k )^2  )^{\frac{1}{2}}=\BO(M^{-\frac{1}{2}}).
\end{aligned}
\end{equation}
Next, we turn to evaluate $X_4$. By using~(\ref{h_Q_QI}), $X_4$ can be evaluated by
\begin{equation}
\begin{aligned}
X_4 &=-\sum_{k=1}^{M}\frac{ \Tr \BT (\E_{j}\Bh_k \Bh_k^{H}\BQ) \BD_{j}\BQ\BD_i}{M}
=-\sum_{k=1}^{M} \frac{ \rho \widetilde{t}_{k,k}}{M} \E\Tr \BT (\E_{j}\Bh_k \Bh_k^{H}\BQ_k) \BD_{j}\BQ\BD_i+o(1)
\\
&
=(-\sum_{k=1}^{M}  \frac{ \rho \widetilde{t}_{k,k}}{M} \E\Tr \BT (\E_{j}\Ba_k \Ba_k^{H}\BQ) \BD_{j}\BQ\BD_i)
+(-\sum_{k=1}^{M} \frac{ \rho \widetilde{t}_{k,k}}{M} \E\Tr \BT (\E_{j}\By_k \By_k^{H}\BQ_k) \BD_{j}\BQ\BD_i)
\\
&
+(-\sum_{k=1}^{M}  \frac{ \rho \widetilde{t}_{k,k}}{M}\E \Tr \BT (\E_{j}\Ba_k \By_k^{H}\BQ_k) \BD_{j}\BQ\BD_i)
+(-\sum_{k=1}^{M} \frac{ \rho \widetilde{t}_{k,k}}{M}  \E\Tr \BT (\E_{j}\By_k \Ba_k^{H}\BQ_k) \BD_{j}\BQ\BD_i)+o(1)
\\
&
=X_{4,1}+X_{4,2}+X_{4,3}+X_{4,4}+o(1).
\end{aligned}
\end{equation}
 By~(\ref{t_jj_exp2}), we can obtain the evaluation for $X_{4,1}$, which can be cancelled by $X_{3,0}+X_{3,1}$ with
\begin{equation}
\label{eva_X41}
\begin{aligned}
X_{4,1} &=-\sum_{k=1}^{M}  \frac{\rho\widetilde{t}_{k,k}\E\Tr \BT (\E_{j}\Ba_k \Ba_k^{H}\BQ) \BD_{j}\BQ\BD_i}{M}
=-\sum_{k=1}^{M}  \frac{\rho \widetilde{t}_{k,k}}{M}  \E\Ba_k^{H}(\E_{j}\BQ) \BD_{j}\BQ\BD_i\BT \Ba_k
\\
&
\overset{(a)}{=}-(X_{3,0}+X_{3,1})+o(1),
\end{aligned}
\end{equation}
where step $(a)$ in~(\ref{eva_X41}) is obtained by~(\ref{t_jj_exp2}). By~(\ref{ide_Q_QI}), we can obtain the evaluation for $X_{4,2}$ as
\begin{equation}
\label{X42_eva}
\begin{aligned}
X_{4,2} & =-\sum_{k=1}^{j} \frac{  \rho\widetilde{t}_{k,k} \E \By_k^{H}(\E_{j}\BQ_k) \BD_{j}\BQ\BD_i\BT \By_k}{M}
-\sum_{k=j+1}^{M}\frac{\E  \rho \widetilde{t}_{k,k} \Tr(\E_{j}\BQ) \BD_{j}\BQ\BD_i\BT\BD_k}{M^2}
\\
&=-\sum_{k=1}^{j}\frac{\E  \rho \widetilde{t}_{k,k}  \By_k^{H}(\E_{j}\BQ_k) \BD_{j}\BQ_k\BD_i\BT \By_k}{M}
+(\sum_{k=1}^{j}\frac{\E  \rho^2 \widetilde{t}_{k,k}^2  \By_k^{H}(\E_{j}\BQ_k) \BD_{j}\BQ_k\Bh_k\Bh_k^{H} \BQ_k\BD_i\BT \By_k}{M}
\\
&
-\sum_{k=j+1}^{M}\frac{\E  \rho\widetilde{t}_{k,k} \Tr(\E_{j}\BQ) \BD_{j}\BQ\BD_i\BT\BD_k}{M^2})+o(1)
\\
&=-\sum_{k=1}^{M}\frac{\E  \rho\widetilde{t}_{k,k} \Tr(\E_{j}\BQ_k) \BD_{j}\BQ_k\BD_i\BT \BD_k}{M^2}
+\sum_{k=1}^{j}\frac{\rho^2 \widetilde{t}_{k,k}^2\E\Tr \BT\BD_i\BT\BD_k \E \Tr (\E_{j}\BQ_k) \BD_{j}\BQ_k\BD_k}{M^3}+o(1)
\\
&\overset{(a)}{=}-\sum_{k=1}^{M}  \frac{\rho\E \Tr\BT \bold{\Upsilon}(\E_{j}\BQ_k) \BD_{j}\BQ_k\BD_i}{M}
+\sum_{k=1}^{j} \rho^2\widetilde{t}_{k,k}^2  \Gamma_{i,k}  \psi_{j,k} +o(1)
=X_{4,2,1}+X_{4,2,2}+o(1),
\end{aligned}
\end{equation}
where step $(a)$ in~(\ref{X42_eva}) follows by $\frac{1}{M}\sum_{k=1}^{M}\sigma^2_{i,k}\widetilde{t}_{k,k}=\frac{\Tr\widetilde{\BD}_i\widetilde{\BT}}{M}=\widetilde{\delta}_i$. By applying~(\ref{ide_Q_QI}) to $\BQ$, we have the following evaluate for $X_{4,3}$ and $X_{4,4}$.
\begin{equation}
\label{X_43}
\begin{aligned}
X_{4,3}&=
\sum_{k=1}^{M}\frac{ \rho^2 \widetilde{t}_{k,k}^2\E\Tr \BT (\E_{j}\Ba_k \By_k^{H}\BQ_k) \BD_{j}\BQ_k\Bh_k\Bh_k^{H} \BQ_k\BD_i}{M}+o(1)
\\
&=\sum_{k=1}^{j} \frac{ \rho^2 \widetilde{t}_{k,k}^2 \Ba_k^{H} \BT_k\BD_i\BT\Ba_k \Tr (\E_{j} \BQ_k) \BD_{j}\BQ_k\BD_k}{M^2}+o(1)
\\
&\overset{(a)}{=}\sum_{k=1}^{j}  \frac{\rho \widetilde{t}_{k,k} \Ba_k^{H} \BT_k\BD_i\BT\Ba_k \psi_{j,k}}{M(1+\delta_k)}+o(1),
\end{aligned}
\end{equation}
where step $(a)$ in~(\ref{X_43}) follows from~(\ref{rho_aTb}) and
\begin{equation}
\begin{aligned}
X_{4,4} &=\sum_{k=1}^{M}\frac{ \rho^2 \widetilde{t}_{k,k}^2}{M} \E\Tr \BT (\E_{j}\By_k \Ba_k^{H}\BQ_k) \BD_{j}\BQ_k\Bh_k\Bh_k^{H} \BQ_k\BD_i+o(1)
\\
&=\sum_{k=1}^{j}  \frac{\rho^2 \widetilde{t}_{k,k}^2 }{M^2}\Tr  \BT\BD_{i}\BT\BD_k \E \Ba_k^{H} (\E_{j} \BQ_k) \BD_j\BQ_k\Ba_k +o(1)
\\
&=\sum_{k=1}^{j}  \rho^2 \widetilde{t}_{k,k}^2 \Gamma_{k,i}\theta_{j,k} +o(1).
\end{aligned}
\end{equation}
Noticing that $X_{4,2,1}+X_{2}=0$ and $X_{4,1}+X_{3,0}+X_{3,1}=0$,~(\ref{psi_first}) can be rewritten as
\begin{equation}
\begin{aligned}
\label{eq_second_pre}
\psi_{j,i}&=X_{1}+(X_{3,3}+X_{4,3}+X_{4,2,2})+X_{4,4}
\\
&=
M\Gamma_{i,j}+ \sum_{i=1}^{j}[ \frac{\Ba_k^{H} \BT\Ba_k \Ba_k^{H}\BT \BD_i\BT\Ba_k}{M(1+\delta_k)^3}
+\frac{\rho\widetilde{t}_{k,k} \Ba_k^{H} \BT_k\BD_i\BT\Ba_k}{(1+\delta_k)}+\rho^2\widetilde{t}_{k,k}^2\Gamma_{i,k}  ]\psi_{j,k}
+\sum_{k=1}^{j}  \rho^2\widetilde{t}_{k,k}^2 \Gamma_{k,i}\theta_{j,k}
\\
&=
\Gamma_{i,j}+ \sum_{k=1}^{j}(\Pi_{i,k} +\rho^2\widetilde{t}_{k,k}^2\Gamma_{i,k}  )\psi_{j,k}
+\sum_{k=1}^{j}  \rho^2 \widetilde{t}_{k,k}^2 \Gamma_{k,i}\theta_{j,k}+o(1).
\end{aligned}
\end{equation}
By the bound for the rank-one perturbation in~(\ref{rank_one_per}), we can obtain $\frac{\E\Tr (\E_j \BQ)\BD_j\BQ\BD_i}{M}=\psi_{j,i}+\BO(M^{-1})$ and rewrite~(\ref{eq_second_pre}) as
\begin{equation}
\begin{aligned}
\label{eq_second}
& \sum_{k=1}^{j}[ \mathbbm{1}_{k=i} -(\Pi_{i,k} +\rho^2\widetilde{t}_{k,k}^2\Gamma_{i,k}  )]\psi_{j,k}
-\sum_{k=1}^{j}  \rho^2 \widetilde{t}_{k,k}^2 \Gamma_{k,i}\theta_{j,k}=M\Gamma_{i,j} +o(1).
\end{aligned}
\end{equation}
Recall $\Theta_{j,k}=\rho^2\widetilde{t}_{k,k}^2 \theta_{j,k}$ and define $\bold{p}_{j}\in \mathbb{R}^{2j}$, $\bold{q}_{j}\in \mathbb{R}^{2j} $ as
\begin{equation}
\begin{aligned}
\bold{p}_{j}&=[\psi_{j,1},\psi_{j,2},...,\psi_{j,j}, \rho^2\widetilde{t}_{1,1}^2 \psi_{j,1}+\Theta_{j,1},
 \rho^2\widetilde{t}_{2,2}^2 \psi_{j,2}+\Theta_{j,2},..., \rho^2\widetilde{t}_{j,j}^2 \psi_{j,j}+\Theta_{j,j}]^{T},
\\
\bold{q}_{j}&=M[\Gamma_{j,1},\Gamma_{j,2},...,\Gamma_{j,j},\Pi_{j,1},\Pi_{j,2},...,\Pi_{j,j}]^{T}.
\end{aligned}
\end{equation}
By~(\ref{eq_first}) and~(\ref{eq_second}), we can set up a system of equations with respect to $\Theta_{j,i}$ and $\psi_{j,i}$ as follows
\begin{equation}
\bold{p}_{j}=\bold{B}_{j}\bold{p}_{j}+\bold{q}_{j}+\bold{v}_{j},
\end{equation}
where $\bold{B}_{j}$ is given in~(\ref{def_BB}) and $| [v_{j}]_{i} |=o(1)$, $i=1,2,...,2j$.
Therefore, we have 
\begin{equation}
\bold{p}_{j}=(\BS_{j})^{-1}(\bold{q}_{j}+\bold{v}_{j}),
\end{equation}
with $\BS_{j}=\bold{I}_{2j}-\BB_j$. We define $\bold{\Omega}_j=\bold{\Xi}_j+\widetilde{\bold{\Lambda}}_{j}$ for the following derivations.
\subsection{Evaluation of $\frac{1}{M}(\sum_{j=1}^{M}-\rho^2\widetilde{t}_{j,j}^2 p_{j,(j)}+2 p_{j,(2j)})$.}
\label{sec_eva_sum_mart}
By Cramer's rule, $p_{j,(j)}$ and $p_{j,(2j)}$ can be represented by
\begin{equation}
\begin{aligned}
\label{gram_eqs}
p_{j,(j)}&=\psi_{j,j}=\frac{\det(\BS_{j,(j)})}{\det(\BS_{j})}+[\BS_{j}^{-1}\bold{v}_{j}]_{(j)},
\\
p_{(j),2j}&=\rho^2\widetilde{t}_{j,j}\psi_{j,j}+\Theta_{j,j}=\frac{\det(\BS_{j,(2j)})}{\det(\BS_{j})}+[\BS_{j}^{-1}\bold{v}_{j}]_{(2j)},
\end{aligned}
\end{equation}
where $\BS_{j,(j)}$ and $\BS_{j,(2j)}$ are obtained by replacing the $j$-th and $2j$-th columns of $\BS_{j}$ with $\bold{q}_{j}$, respectively. 
Now we will show that 
\begin{equation}
-\rho^2\widetilde{t}_{j,j}^2 p_{j,(j)}+2 p_{j,(2j)}
=\frac{\det(\BS_{j-1})-\det(\BS_{j})}{\det(\BS_{j})}+o(1).
\end{equation}
First, by adding $\rho^2\widetilde{t}_{j,j}^2$ times of the $2j$-th column to the $j$-th column $\BS_{j,(j)}$, we can obtain
\begin{equation}
\label{psi_jj_first}
\small
\begin{aligned}
&\frac{\rho^2\widetilde{t}_{j,j}^2}{M} \det( \BS_{j,(j)})
=- \det(
\begin{bmatrix}
 \bold{I}_{j-1}-\bold{\Pi}_{j-1} & \bold{0} & -\bold{\Gamma}_{j}^{(1:j-1)}
 \\
 -\bold{\Pi}_{j}^{(j)}& 0 & -\bold{\Gamma}_{j}^{(j)}
 \\
-\bold{\Omega}_{j-1} & \bold{0} &\bold{I}_{j}^{(1:j-1)}-\bold{\Pi}_{(j)}^{[1:1-j],T}
\\
-\bold{\Xi}_{j}^{(j)}  & -\rho^2\widetilde{t}_{j,j}^2   &   \bold{e}_{j}^{T}-\bold{\Pi}_{(j)}^{[j],T}
 \end{bmatrix}
    ),
    \end{aligned}
\end{equation}
without changing $\det(\BS_{j,(j)})$, where $\bold{e}_{j}\in\mathbb{R}^{j}=\bold{I}_{j}^{[j]} =(0,0,...,1)^{T}$. 

\begin{figure*}
\begin{equation}
\begin{aligned}
\small
\nonumber
&\frac{\det( \BS_{j,(2j)})}{M}=
\\
&
 -\det(
\begin{bmatrix}
 \bold{I}_{j-1}-\bold{\Pi}_{j-1} & -\bold{\Pi}_{j}^{[j]} & -\bold{\Gamma}_{j}^{(1:j-1)}
 \\
 -\bold{\Pi}_{j}^{(j)}& 1-\Pi_{j,j} & -\bold{\Gamma}_{j}^{(j)}
 \\
-\bold{\Xi}_{j-1}-\widetilde{\bold{\Lambda}}_{j-1} &-\bold{\Xi}_{j}^{[j]} &\bold{I}_{j}^{(1:j-1)}-\bold{\Pi}_{(j)}^{[1:1-j],T}
\\
-\bold{\Xi}_{j}^{(j)}  & -\rho^2\widetilde{t}_{j,j}^2   &   -\bold{\Pi}_{j}^{[j],T}
 \end{bmatrix})
\!=\!-\det(
\begin{bmatrix}
 \bold{I}_{j-1}-\bold{\Pi}_{j-1} & -\bold{\Pi}_{j}^{[j]} & -\bold{\Gamma}_{j}^{(1:j-1)}
 \\
 -\bold{\Pi}_{j}^{(j)}& 1-\Pi_{j,j} & -\bold{\Gamma}_{j}^{(j)}
 \\
-\bold{\Xi}_{j-1}-\widetilde{\bold{\Lambda}}_{j-1} &-\bold{\Xi}_{j}^{[j]} &\bold{I}_{j}^{(1:j-1)}-\bold{\Pi}_{(j)}^{[1:1-j],T}
\\
-\bold{\Xi}_{j}^{(j)}  & -\rho^2\widetilde{t}_{j,j}^2   &   \bold{e}_{j}^{T}-\bold{\Pi}_{j}^{[j],T}
 \end{bmatrix})
\end{aligned}
\end{equation}
 \begin{equation}
 \small
 \begin{aligned}
 \label{BA_2j_first}
 & +\det(\begin{bmatrix}
 \bold{I}_{j-1}-\bold{\Pi}_{j-1} & -\bold{\Pi}_{j}^{[j]} & -\bold{\Gamma}_{j}^{(1:j-1)}
 \\
 -\bold{\Pi}_{j}^{(j)}& 1-\Pi_{j,j} & -\bold{\Gamma}_{j}^{(j)}
 \\
-\bold{\Xi}_{j-1}-\widetilde{\bold{\Lambda}}_{j-1} &-\bold{\Xi}_{j}^{[j]} &\bold{I}_{j}^{(1:j-1)}-\bold{\Pi}_{j}^{[1:1-j],T}
\\
\bold{0}& 0   &   \bold{e}_{j}^{T}
 \end{bmatrix})
 \overset{(a)}{=}-\det(
\begin{bmatrix}
 \bold{I}_{j-1}-\bold{\Pi}_{j-1} & -\bold{\Pi}_{j}^{[j]} & -\bold{\Gamma}_{j}^{(1:j-1)}
 \\
 -\bold{\Pi}_{j}^{(j)}& 1-\Pi_{j,j} & -\bold{\Gamma}_{j}^{(j)}
 \\
-\bold{\Xi}_{j-1} -\widetilde{\bold{\Lambda}}_{j}&-\bold{\Xi}_{j}^{[j]} &\bold{I}_{j}^{(1:j-1)}-\bold{\Pi}_{j}^{[1:1-j],T}
\\
-\bold{\Xi}_{j}^{(j)}  & -\rho^2\widetilde{t}_{j,j}^2   &   \bold{e}_{j}^{T}-\bold{\Pi}_{j}^{[j],T}
 \end{bmatrix})
 \\
 &
 +\det(
\begin{bmatrix}
 \bold{I}_{j-1}-\bold{\Pi}_{j-1} & \bold{0} & -\bold{\Gamma}_{j}^{(1:j-1)}
 \\
 -\bold{\Pi}_{j}^{(j)}& 1 & -\bold{\Gamma}_{j}^{(j)}
 \\
-\bold{\Xi}_{j-1} -\widetilde{\bold{\Lambda}}_{j-1}&
\bold{0}
&\bold{I}_{j}^{(1:j-1)}-\bold{\Pi}_{j}^{[1:1-j],T}
\\
\bold{0}
& 0   &    \bold{e}_{j}^{T}
 \end{bmatrix})
+ \det(
\begin{bmatrix}
 \bold{I}_{j-1}-\bold{\Pi}_{j-1} & -\bold{\Pi}_{j}^{[j]} & -\bold{\Gamma}_{j}^{(1:j-1)}
 \\
 -\bold{\Pi}_{j}^{(j)}& -\Pi_{j,j} & -\bold{\Gamma}_{j}^{(j)}
 \\
-\bold{\Xi}_{j-1}-\widetilde{\bold{\Lambda}}_{j-1} &-\bold{\Xi}_{j}^{[j]} &\bold{I}_{j}^{(1:j-1)}-\bold{\Pi}_{j}^{[1:1-j],T}
\\
\bold{0}  & 0   &    \bold{e}_{j}^{T}
 \end{bmatrix})
 \\
 &=- \det(\BS_{j})+\det(\BS_{j-1})   +C_1.
\end{aligned}
\end{equation}
\hrulefill
\end{figure*}
On one hand, we can decompose $\frac{\det( \BS_{j,(2j)})}{M}$ as $\frac{\det( \BS_{j,(2j)})}{M}=- \det(\BS_{j})+\det(\BS_{j-1})   +C_1$ in~(\ref{BA_2j_first}) at the top of the next page, where step $(a)$ in~(\ref{BA_2j_first}) is obtained by decomposing the $j$-the column of the second term in the previous step. On the other hand, we can decompose $\frac{\det( \BS_{j,(2j)})}{M}$ as $\frac{\det( \BS_{j,(2j)})}{M}=\frac{\rho^2\widetilde{t}_{j,j}^2}{M} \det( \BS_{j,(j)})-C_1-C_2$ in~(\ref{BA_2j_second}) on next page,
\begin{figure*}
\begin{equation}
\small
\label{BA_2j_second}
\begin{aligned}
\frac{\det( \BS_{j,(2j)}) }{M} &\overset{(a)}{=} -(-1)^{j} \det(\begin{bmatrix}
-\bold{\Gamma}_{j}^{(1:j-1)}&   \bold{I}_{j-1}\!-\!\bold{\Pi}_{j-1} & -\bold{\Pi}_{j}^{[j]}
 \\
-  \bold{\Gamma}_{j}^{(j)} &  -\bold{\Pi}_{j}^{(j)}&   1-\Pi_{j,j} 
 \\
\bold{I}_{j}^{(1:j-1)}\!-\!\bold{\Pi}_{j}^{[1:1-j],T}&-\bold{\Xi}_{j-1}-\widetilde{\bold{\Lambda}}_{j} &  -\bold{\Xi}_{j}^{[j]} 
\\
-\bold{\Pi}_{j}^{[j],T} &-\bold{\Xi}_{j}^{(j)} & -\rho^2\widetilde{t}_{j,j}^2
\end{bmatrix})
\\
&
\overset{(b)}{=}-(-1)^{j} \det(\begin{bmatrix}
-\bold{\Gamma}_{j}^{[1:j-1]}&  -  \bold{\Gamma}_{j}^{[j]} &\bold{I}_{j}^{(1:j-1)}\!-\!\bold{\Pi}_{j}^{[1:1-j]} & -\bold{\Pi}_{j}^{[j]}
 \\
 \bold{I}_{j-1}\!-\!\bold{\Pi}_{j-1}^{T} &  -\bold{\Pi}_{j}^{(j),T}&  -\bold{\Xi}_{j-1} -\widetilde{\bold{\Lambda}}_{j} & -\bold{\Xi}_{j}^{(1:j-1),[j]}
 \\
 -\bold{\Pi}_{j}^{[j],T}&1-\Pi_{j,j} &  -\bold{\Xi}_{j}^{(j),[1:j-1]} & -\rho^2\widetilde{t}_{j,j}^2
\end{bmatrix})
\\
&
\overset{(c)}{=}-\det(\begin{bmatrix}
\bold{I}_{j}^{(1:j-1)}\!-\!\bold{\Pi}_{j}^{[1:1-j]} & -\bold{\Pi}_{j}^{[j]}&-\bold{\Gamma}_{j}^{[1:j-1]}&  -  \bold{\Gamma}_{j}^{[j]} 
 \\
  -\bold{\Xi}_{j-1} -\widetilde{\bold{\Lambda}}_{j} & -\bold{\Xi}_{j}^{[j]}&  \bold{I}_{j-1}\!-\!\bold{\Pi}_{j-1}^{T} &  -\bold{\Pi}_{j}^{(j),T}
 \\
  -\bold{\Xi}_{j}^{(j)} & -\rho^2\widetilde{t}_{j,j}^2&  -\bold{\Pi}_{j}^{[j],T}&1-\Pi_{j,j}
\end{bmatrix})
=\frac{\rho^2\widetilde{t}_{j,j}^2}{M} \det( \BS_{j,(j)})-C_1-C_2.
\end{aligned}
\end{equation}
\hrulefill
\end{figure*}
where 
\begin{equation}
\small
\begin{aligned}
&C_2=\det(
\begin{bmatrix}
\bold{I}_{j}^{[1:j-1]}\!-\!\bold{\Pi}_{j}^{[1:1-j]} & -\bold{\Pi}_{j}^{[j]}&-\bold{\Gamma}_{j}^{[1:j-1]}&  -  \bold{\Gamma}_{j}^{[j]} 
 \\
  -\bold{\Xi}_{j-1}-\widetilde{\bold{\Lambda}}_{j-1}  & -\bold{\Xi}_{j}^{[j]}&  \bold{I}_{j-1}\!-\!\bold{\Pi}_{j-1}^{T} &  -\bold{\Pi}_{j}^{(j),T}
 \\
  -\bold{\Xi}_{j}^{(j)} & 0&  -\bold{\Pi}_{j}^{[j],T}&-\Pi_{j,j}
\end{bmatrix}).
\end{aligned}
\end{equation}
Step $(a)$ in~(\ref{BA_2j_second}) is obtained by exchanging the $i$-th column with $j+i$-th column of the determinant in the first line of~(\ref{BA_2j_first}), step (b) follows from the identity $|\bold{M}|=|\bold{M}^{T}|$, and step $(c)$ follows by exchanging $i$-th column with $(j+i)$-th column. 
By~(\ref{psi_jj_first}), ~(\ref{BA_2j_first}), and~(\ref{BA_2j_second}), we can obtain
\begin{equation}
\frac{\rho^2\widetilde{t}_{j,j}^2 \det(\BS_{j,(j)})}{M}+\frac{ 2\det( \BS_{j,(2j)})}{M}
=\det(\BS_{j-1})-\det(\BS_j)-C_2.
\end{equation}
Next, we will show that
\begin{equation}
C_2=\BO(M^{-2}).
\end{equation}
By moving the $j$-th column to the $(2j-1)$-th column and the $j$-th row to the $(2j-1)$-th row, $C_2$ can be rewritten as
\begin{equation}
\label{C_3_eva}
\begin{aligned}
C_2=\det(\begin{bmatrix}
\BS_{j-1} &\bold{G}_j\\
\widetilde{\bold{G}}_{j}&\bold{E}_j
\end{bmatrix}),
\end{aligned}
\end{equation}
where
\begin{equation}
\label{g_def}
\begin{aligned}
\bold{G}_j \!\!&=-\!\!
\begin{bmatrix}
\bold{\Pi}_{j}^{(1:j-1),[j]} & \bold{\Gamma}_{j}^{(1:j-1),[j]}
\\
\bold{\Xi}_{j}^{(1:j-1),[j]}  & \bold{\Pi}_{j}^{(1:j-1),[j]}
\end{bmatrix},
\bold{E}_j \!\!=\!\!
-\begin{bmatrix}
\Pi_{j,j} & \Gamma_{j,j}
\\
0 & \Pi_{j,j}
\end{bmatrix}
\\
\widetilde{\bold{G}}_j \!\!&=-\!\!
\begin{bmatrix}
\bold{\Pi}_{j}^{(j),[1:j-1]} & \bold{\Gamma}_{j}^{(j),[1:j-1]}
\\
\bold{\Xi}_{j}^{(j),[1:j-1]}  & \bold{\Pi}_{j}^{(j),[1:j-1]}
\end{bmatrix}
.
\end{aligned}
\end{equation}
By the trace inequality~(\ref{trace_inq}), we have 
\begin{equation}
\label{lbnd_gamma}
\begin{aligned}
&\frac{\Tr\BD_i \sigma_{min}^2}{M} \le \frac{\Tr\BD_i\BD_j}{M}=\frac{\Tr\BD_i\FT\FT^{-1}\BD_j}{M}
\le \frac{\Tr\FT^{\frac{1}{2}}\BD_j\BD_i\FT^{\frac{1}{2}}}{M} \|\FT^{-1} \| \le \frac{\Tr\BD_i\FT\BD_j\FT}{M} \|\FT^{-1} \|^2.
\end{aligned}
\end{equation}
According to~(\ref{lbnd_gamma}) and~\textbf{A.2}, we can obtain $\inf\limits_{i,j} M\Gamma_{i,j} >0$. Similarly, we have
\begin{equation}
\begin{aligned}
&\sigma^2_{min}\Ba_i^{H}\Ba_i =\sigma^2_{min}\Tr\Ba_i \Ba_i^{H}\FT\FT^{-1}
\le \sigma^2_{min}\Tr\Ba_i \Ba_i^{H}\FT \| \FT^{-1} \|
\le\sigma^2_{min} \Tr\Ba_i \Ba_i^{H}\FT \| \FT^{-1} \|^2 \le \Ba_i^{H}\FT\BD_{j} \FT\Ba_i\| \FT^{-1} \|^2,
\end{aligned}
\end{equation}
such that $\inf\limits_{i,j}M \Pi_{i,j}>0$. Therefore, we can obtain $\inf\limits_{i \le j} [\bold{q}_{j}]_i >0$. By~\cite[Proposition 5.5, Lemma 5.1]{hachem2008clt}, we have: (i). $\bold{S}_j$ is invertible; (ii). $\det(\BS_j)$ is bounded; (iii). the $1$-norm for the rows of $\BS_j^{-1}$ and $[\BS_j^{(p:j),[p:j]})]^{-1}$ are bounded, i.e., $\max_{i\le j} \| [\bold{S}_{j}^{-1}]^{(i)}\|_{1} < K_s $ and $\max_{p\le j}\| [\BS_j^{(p:j),[p:j]})]^{-1} \|_{1} <K_s$, where $K_s$ is independent of $M$, $N$, and $j$; and (iv) $[\bold{S}_{j}^{-1}]_{m,n}>0$. We also have $|[\BE]_{i,j}|\le K_e M^{-1}$. Therefore, for $m=1$ or $n=2$, $|[\BG]_{i,j}|\le K_{g} M^{-1}$, there holds
\begin{equation}
\begin{aligned}
\label{G_other}
&[\widetilde{\bold{G}}_j\BS_{j-1}^{-1}\bold{G}_j]_{m,n} \le  \|\widetilde{\BG}_{j}^{(m)} \|_1 \max_{i\le j} \| [\bold{S}_{j}^{-1}]^{(i)}\|_{1}
 \| \bold{G}_{j}^{[n]} \|_{0} 
\le  K_g^2 K_e M^{-1}. 
\end{aligned}
\end{equation}
By~(\ref{g_def}), $[\widetilde{\bold{G}}_j\BS_{j-1}^{-1}\bold{G}_j]_{2,1}$  can be evaluated by 
\begin{equation}
\begin{aligned}
&[\widetilde{\bold{G}}_j\BS_{j-1}^{-1}\bold{G}_j]_{2,1}=\bold{\Xi}_{j}^{(j),[1:j-1]}[\BS_{j-1}^{-1}]^{(1,2)}\bold{\Xi}_{j}^{(1:j-1),[j]}
+\bold{\Pi}_{j}^{(j),[1:j-1]}[\BS_{j-1}^{-1}]^{(2,1)}\bold{\Pi}_{j}^{(1:j-1),[j]}
\\
&
+\bold{\Xi}_{j}^{(j),[1:j-1]}   [\BS_{j-1}^{-1}]^{(1,1)}\bold{\Pi}_j^{(1:j-1),[j]}.
+\bold{\Pi}_{j}^{(j),[1:j-1]}   [\BS_{j-1}^{-1}]^{(1,1)}\bold{\Xi}_{j}^{(1:j-1),[j]}
\\
&
=E_1+E_2+E_3+E_4,
\end{aligned}
\end{equation}
where $[\bold{S}_{j-1}^{-1}]^{(m,n)}$ represents the $(m,n)$-th $(j-1)\times (j-1)$ block of $[\bold{S}_{j-1}^{-1}]$, $m,n=1,2$. By similar analysis as~(\ref{G_other}), we can show that $E_2$, $E_3$, and $E_4$ are $\BO(M^{-1})$. Next we show that $E_1$ is also $\BO(M^{-1})$. By the block matrix inversion formula~(\ref{block_inv}) in Appendix~\ref{appendix_lemmas}, $E_1$ can be computed by
\begin{equation}
\begin{aligned}
&[\bold{S}_{j-1}^{-1}]^{(1,2)}=[\bold{I}_{j-1}-\bold{\Pi}_{j-1} -\bold{\Gamma}_{j-1}
(\bold{I}_{j-1}-\bold{\Pi}_{j-1}^{T})^{-1}\bold{\Xi}_{j-1} ]^{-1}\bold{\Gamma}_{j-1}(\bold{I}_{j-1}-\bold{\Pi}_{j-1}^{T})^{-1},
\end{aligned}
\end{equation}
so that the $(m,n)$-th entry of $[\bold{S}_{j-1}^{-1}]^{(1,2)}$ is bounded by
\begin{equation}
\begin{aligned}
\label{S_j_1_bnd}
&|[\bold{S}_{j-1}^{-1}]^{(1,2)}_{m,n}| \le  \max_{a,b}| [\bold{\Gamma}_{j-1}]_{a,b} |    \| (\bold{I}_{j-1}-\bold{\Pi}_{j-1} )^{-1}_{(n)}\|_1
 \| (\bold{I}_{j-1}-\bold{\Pi}_{j-1}-\bold{\Xi}_{j-1} (\bold{I}_{j-1}-\bold{\Pi}_{j-1}^{T})^{-1} \bold{\Gamma}_{j-1})^{-1}_{(m)} \|_{1} 
  \le \frac{K}{M},
\end{aligned}
\end{equation}
where the boundness of the two $1$-norms in~(\ref{S_j_1_bnd}) follows from boundness in (iii) above~(\ref{G_other}) as
\begin{equation}
 \begin{aligned}
& \| (\bold{I}_{j-1}-\bold{\Pi}_{j-1}-\bold{\Xi}_{j-1} (\bold{I}_{j-1}-\bold{\Pi}_{j-1}^{T})^{-1} \bold{\Gamma}_{j-1})^{-1}_{(m)} \|_{1} 
 \| [\bold{S}_{j}^{-1}]^{(1,1)}_{(m)} \|_{1} \le \| [\bold{S}_{j}^{-1}]^{(m)} \|_{1} < \infty.
 \end{aligned}
 \end{equation}
According to the boundness of $\| \bold{\Xi}_{j,(j)}\|_{1} $ in
\begin{equation}
\| \bold{\Xi}_{j}^{(j)}\|_{1} \le K \Tr\BA\BT\Ba_j\Ba_j^{H}\BA \le K a_{max}^4\rho^{-2}, 
\end{equation}
we can obtain
\begin{equation}
\begin{aligned}
&| \bold{\Pi}_{j}^{(j),T}[\bold{S}_{j-1}^{-1}]_{(1,2)} \bold{\Pi}_{j}^{(j)}| \le \| \bold{\Pi}_{j}^{(j)}  \|_{1}^2 \max_{a,b} | [\bold{S}_{j-1}^{-1}]^{(1,2)}_{a,b}  | 
= \BO(M^{-1}).
\end{aligned}
\end{equation}
By the determinant formula for block-matrix with invertible $\BA$, i.e.,
\begin{equation}
\begin{aligned}
\det(\begin{bmatrix}
\BA &\BB\\
\BC&\BD
\end{bmatrix})
=\det(\BA)\det(\BD-\BB\BA^{-1}\BD),
\end{aligned}
\end{equation}
we can obtain the following evaluation for $|C_2|$
\begin{equation}
\begin{aligned}
&|C_2|=|\det(\BS_{j-1}) | | \det( \bold{E}_j - \widetilde{\bold{G}}_j\BS_{j-1}^{-1}\bold{G}_j) | 
\le  K  ( K_g M^{-1}+ K_e K_g^2  M^{-1}  )^2 =\BO(M^{-2}),
\end{aligned}
\end{equation}
which proves~(\ref{C_3_eva}). By $\max_{i\le j} \| [\bold{S}_{j}^{-1}]^{(i)}\|_{1} < K_s $, there holds true that $  \|\BS_{j}^{-1}\bold{v}_{(j)}\|_{0}=o(1) $. Therefore, we have
\begin{equation}
\begin{aligned}
\frac{1}{M}(\sum_{j=1}^{M}\rho^2\widetilde{t}_{j,j}^2 p_{j,(j)}+2 p_{j,(2j)})
&=\sum_{j=1}^{M}\frac{\det(\BS_{j-1})-\det(\BS_{j})}{\det(\BS_{j})} + o(1)
\\
&
=\sum_{j=1}^{M}\log\left(1+\frac{\det(\BS_{j-1})-\det(\BS_{j})}{\det(\BS_{j})}\right) + o(1)
\\
&
=-\log\det(\BS_{M})+ o(1)
\\
&
=-\log\det(\bold{I}_{2j}-\bold{B}_j)+o(1),
\end{aligned}
\end{equation}
which concludes~(\ref{sum_pjj}).

\section{Useful Results}
\label{appendix_lemmas}
\textit{1.  Expansion of the covariance of two quadratic forms, Eq. (3.20) in~\cite{hachem2012clt}.}
\begin{lemma}
\label{lemma_ext}
 Define $\bold{z}=\bold{a}+ \frac{1}{\sqrt{N}}\bold{D}^{\frac{1}{2}}\bold{x} $, where $\bold{x}=(X_{1},X_{2},...,X_{N})^{T}$, which is a random vector with i.i.d circularly Gaussian entries with unit variance, $\bold{D}$ is a diagonal non-negative matrix and $\bold{a} \in \mathbb{C}^{N}$ is a deterministic vector. Assuming that $\bold{\Gamma}  $ and $\bold{\Lambda}$ are two $N\times N$ deterministic matrices, the covariance of the quadratic forms $\bold{z}^{H}\bold{\Gamma}\bold{z}$ and $\bold{z}^{H}\bold{\Lambda}\bold{z}$ is given by
\begin{equation}
\label{qua_ext}
\begin{aligned}
& \E (\bold{z}^{H}\bold{\Gamma}\bold{z} - \E \bold{z}^{H}\bold{\Gamma}\bold{z}  ) 
(\bold{z}^{H}\bold{\Lambda}\bold{z} - \E \bold{z}^{H}\bold{\Lambda}\bold{z}  )  
\\
&=\frac{1}{N^2}\Tr \bold{\Gamma}\bold{D}\bold{\Lambda}\bold{D}
+\frac{1}{N} \bold{a}^{H}\bold{\Gamma}\bold{D}\bold{\Lambda}\bold{a}
+\frac{1}{N} \bold{a}^{H}\bold{\Lambda}\bold{D}\bold{\Gamma}\bold{a}.
\end{aligned}
\end{equation}
\end{lemma}

\textit{2. Approximations for the bilinear form of resolvents.}
\begin{lemma} 
\label{appro_diag_ele}
Given assumptions~\textbf{A.1}-\textbf{A.3}, $p\in [1,2]$, and $\|\bold{U}\|, \| \bold{u}\|_{2}, \| \bold{v}\|_{2}<\infty$, there holds true that
\begin{subequations}
\label{bilinear_approx}
\begin{align}
\frac{\E\Tr\bold{U}\BQ}{M}-\frac{\Tr\bold{U}\FT}{M}=\varepsilon_{1,M}\label{utraceq},
\\
\E | \bold{u}^{H}(\BQ-\BT)\bold{v}  |^{2p} =\varepsilon_{2,M}\label{uqvcon},
\\
\E | \bold{u}^{H}(\BQ_i-\BT_i)\bold{v}  |^{2p} =\varepsilon_{3,M}\label{uqvcon_},
\end{align}
\end{subequations}
where $\varepsilon_{1,n} \xrightarrow{M\rightarrow \infty} 0$, $\varepsilon_{2,n} \xrightarrow{M\rightarrow \infty} 0$, and $\varepsilon_{3,n} \xrightarrow{M\rightarrow \infty} 0$. In particular, 
\begin{equation}
\label{qjj_approx}
\E | \widetilde{q}_{j,j}-\widetilde{t}_{j,j} |^{2p} =o(1).
\end{equation}
\end{lemma}
(\ref{utraceq}),~(\ref{uqvcon}), and~(\ref{uqvcon}) can be obtained by a similar approach as~\cite{hachem2007deterministic,hachem2013bilinear}.

\textit{3. Rank-one perturbation~\cite[Lemma 3.1]{hachem2012clt}.} For any matrix $\BA$, the resolvent $\BQ$ and the rank-one perturbation resolvent $\BQ_i$ satisfy
\begin{equation}
\label{rank_one_per}
| \Tr\BA(\BQ-\BQ_i) | \le \frac{\| \BA\|}{-z}.
\end{equation}

\textit{4. Block matrix inverse formula.} For square matrices $\BA$ and $\BB$, the inversion of the block matrix $\begin{bmatrix}
\BA & \BB \\
\BC &\BD
\end{bmatrix}$ is given by~(\ref{block_inv}) at the top of the next page.
\begin{figure*}[t!]
\begin{equation}
\label{block_inv}
\begin{bmatrix}
\BA & \BB \\
\BC &\BD
\end{bmatrix}^{-1}
=\begin{bmatrix}
\left( \BA-\BB\BD^{-1}\BC \right)^{-1} & -\left( \BA-\BB\BD^{-1}\BC \right)^{-1}\BB\BD^{-1}\\
-\BD^{-1}\BC\left( \BA-\BB\BD^{-1}\BC \right)^{-1} & \BD^{-1}+\BD^{-1}\BC \left( \BA-\BB\BD^{-1}\BC \right)^{-1}\BB\BD^{-1}
\end{bmatrix}.
\end{equation}
\hrulefill
\end{figure*}

5. Trace inequality. If $\BA$ is a non-negative matrix, we have
\begin{equation}
\label{trace_inq}
|\Tr\BA\BB|\le \|\BB \|\Tr\BA.
\end{equation}

\ifCLASSOPTIONcaptionsoff
  \newpage
\fi



%
\bibliographystyle{IEEEtran}
\bibliography{IEEEabrv,ref}
%




\end{document}